\documentclass[11pt]{article}

\usepackage[margin=1in]{geometry}
\usepackage{amsmath,amssymb,amsthm}
\usepackage{bm}
\usepackage{graphicx}
\usepackage{natbib}
\usepackage{algorithm}
\usepackage{algpseudocode}
\usepackage{float}
\usepackage{caption}
\usepackage[hidelinks]{hyperref}

\graphicspath{{pic/}}

\theoremstyle{plain}
\newtheorem{theorem}{Theorem}
\newtheorem{proposition}{Proposition}
\newtheorem{lemma}{Lemma}
\newtheorem{corollary}{Corollary}
\theoremstyle{remark}
\newtheorem{example}{Example}

\theoremstyle{definition}

\begin{document}

\title{Manifold Regression}

\author{Yuanyuan Lin\thanks{Corresponding author. Email: \href{mailto:linyuanyuan@amss.ac.cn}{linyuanyuan@amss.ac.cn}.} ~and Dennis K.J. Lin \\[0.5em]
Department of Statistics, Purdue University, West Lafayette, IN 47907, USA}

\date{\today}

\maketitle

\begin{abstract}
  Conventional statistical modeling typically assumes a one-to-one mapping between input and output variables, where ``one-to-one'' is used in the predictive sense that a specified input results in a unique output.
  Many scientific and engineering systems, however, exhibit non-one-to-one (noto) input-output relations.
  In a noto relation, the same input may correspond to multiple admissible outputs, so the conventional statistical models may not be appropriate.
  This paper develops \emph{manifold regression}, a parametric modeling framework for such noto input-output relations, which represents the underlying input-output relation through a latent manifold.
  The manifold is specified upto an unknown finite dimensional parameter and we estimate the unknown parameters from ordinary input-output observations by regularized profile optimization with a robust solver.
  Once the manifold has been learned, prediction is obtained by slicing the estimated manifold along a specified coordinate value.
  Because a slice may contain multiple latent roots, the resulting \emph{manifold prediction} is naturally set-valued.
  This formulation extends predictive learning beyond conventional one-to-one statistical models while containing classical regression and inverse prediction as special cases.
  Theoretical results are established to justify the manifold regression model and its sliced prediction sets;
  and case studies are used to demonstrate stable performance across representative noto cases.
\end{abstract}


\section{Introduction}
\label{sec:intro}

Conventional statistical modeling is shown to be powerful for one-to-one mappings between inputs $\bm{x}$ and outputs $\bm{y}$.
Namely, the identical input $\bm{x}$ results in the same output $\bm{y}$.
Meanwhile, non-one-to-one (noto) input-output relations, where a specified input value $\bm{x}$ may correspond to multiple output values, also arise frequently in scientific and engineering systems.
Consequently, the conventional conditional mean $\mathbb E(\bm{Y}|\bm{X}=\bm{x})$ may not be well defined as a representation of the underlying predictive relation.
Consider the data plotted in Figure~\ref{fig:spiral_intro} as an illustration (the original data are displayed in the Supplementary Material).
Although the scatter exhibits a clear one-dimensional structure, conventional predictive modeling treats the multiple admissible outputs associated with the same input as error-driven variation.
The resulting fitted model averages across these outputs and produces a point-valued prediction, which may not correspond to any meaningful state of the underlying system.

\begin{figure}[!b]
  \centering
  \includegraphics[width=0.38\textwidth]{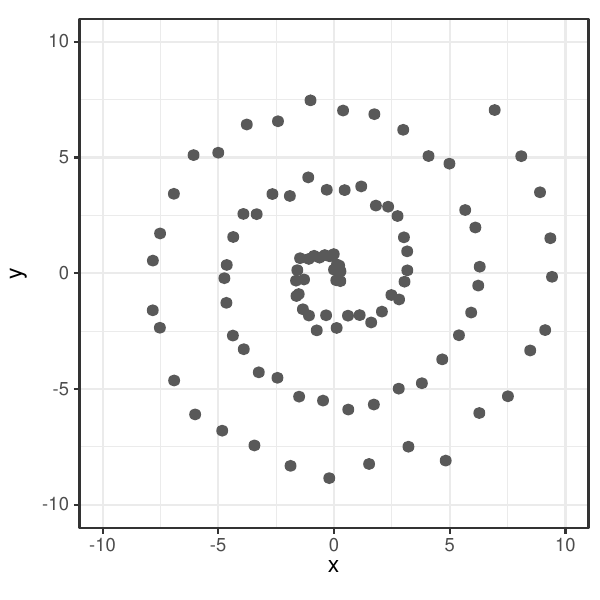}
  \caption{Raw data exhibiting a noto relation between input $x$ and output $y$.}
  \label{fig:spiral_intro}
\end{figure}

Figure~\ref{fig:spiral_intro} makes this noto relation visible, but the issue is much broader.
Hysteresis systems provide a typical physical example: the same input level can produce different outputs depending on the latent history of the system \citep{mayergoyz2003mathematical}.
In robotic kinematics, one end-effector position may be achieved by multiple valid joint configurations, such as elbow-up and elbow-down solutions \citep{spong2005robot}.
Similar ambiguity also appears in modern conditional image generation, where a single conditioning input, such as a partially observed image or a text description, may admit multiple visually distinct but equally valid images \citep{zhu2017toward,rombach2022high}.
These examples share a common statistical feature: the object of interest is not a point-valued prediction, but a set of admissible solutions.

This work shifts the formulation from direct input-output mapping to a fundamentally new perspective.
The key idea is that a noto relation in the input-output coordinates can be represented by a manifold, on which both the inputs and outputs are well-defined mappings of the latent coordinates.
This work develops this idea into a unified framework, called \emph{manifold regression},
where we consider the parametric setting in which the latent mappings have specified functional forms indexed by an unknown finite dimensional parameters.
For each observation, the input is used to infer the feasible 
latent coordinates that match the observed input; among these 
input-matching candidates, the coordinate that best explains the 
observed output is profiled out.
The manifold shape parameters and the active latent domain are then estimated by regularized profile optimization, where the regularization term prevents manifold overfitting. 
When the underlying functional form is unknown, nonparametric estimation with unspecified latent mappings is currently under investigation.
Once this manifold is estimated, prediction becomes a downstream problem, and we slice the manifold at a specified input value and map every retained latent root to its output value.

Manifold regression is geometrically related to manifold learning methods, which recover low-dimensional structure from observations in an ambient space through intrinsic representation or manifold reconstruction
\citep{HastieStuetzle1989,KeglKrzyzakLinderZeger2000,wang2025augmentation,dunson2026inferring}.
However, the two frameworks have very different statistical focuses and methodologies.
Manifold learning generally treats the observed coordinates as components of a common ambient vector, and the learned geometry supports dimension reduction, denoising, and visualization.
Manifold regression focuses on recovering a relation manifold while preserving the distinct roles of input and output throughout the modeling.
Moreover, our manifold regression allows us to tackle the set-valued prediction problem in a principled way.

This work contributes a new statistical formulation of predictive learning. 
At the level of problem formulation, the predictive target is extended beyond the one-to-one case to set-valued prediction. 
At the methodological level, a manifold perspective is introduced into regression, but it preserves the statistical roles of input and output.
Most importantly, the proposed framework broadens the statistical scope of predictive learning while retaining classical regression as a special case. 
The accompanying theory (as will be shown) is thus stated at the level of the objects targeted by the framework, including consistency of manifold modeling and Hausdorff consistency of the resulting prediction sets.

The rest of the paper is organized as follows. 
Section~\ref{sec:illustration} uses the spiral example to motivate the manifold view. 
Section~\ref{sec:modeling} introduces the problem formulation and proposes manifold modeling. 
Relevant theoretical supports are also established.
Section~\ref{sec:prediction} introduces manifold prediction by slicing the estimated manifold, 
Section~\ref{sec:results} presents representative noto case studies, 
and Section~\ref{sec:conclusion} concludes with broader implications and future directions.

\section{An Illustrative Example}
\label{sec:illustration}

This section continues to use the (Archimedean) spiral mentioned in Section~\ref{sec:intro} to motivate the manifold perspective, and demonstrate how this geometric viewpoint naturally unfolds a challenging noto predictive learning problem.

Conventional statistical models represent the relation between input $x$ and output $y$ by a mapping \(y=f(x)+\epsilon\), where $f$ is an unknown function and $\epsilon$ is observation noise. 
Figure~\ref{fig:poly-spline}A and B plot the regression and spline fits against the spiral samples as displayed in Figure~\ref{fig:spiral_intro}. 
Solid, dashed, and dotted lines depict the fitted models, while the points represent the data. 
These plots illustrate why this functional modeling assumption fails for noto relations. 
The first-, second-, and third-degree polynomial model fittings, as displayed in Figure~\ref{fig:poly-spline}A, can only compromise among multiple admissible output values, tending to pass through the empty center of the data rather than follow the observed geometric structure.
The spline fits, including interpolating and smoothing splines, as displayed in Figure~\ref{fig:poly-spline}B, offer much greater flexibility, but flexibility alone does not resolve the problem.
Specifically, interpolating splines oscillate erratically between distant points, while smoothing splines simply average across disjoint vertical branches. 
By contrast, the proposed manifold fit in Figure~\ref{fig:poly-spline}C, as will be discussed, smoothly traces the underlying spiral structure and captures the observed data pattern.

\begin{figure}[tb]
  \centering
  \includegraphics[width=1\textwidth]{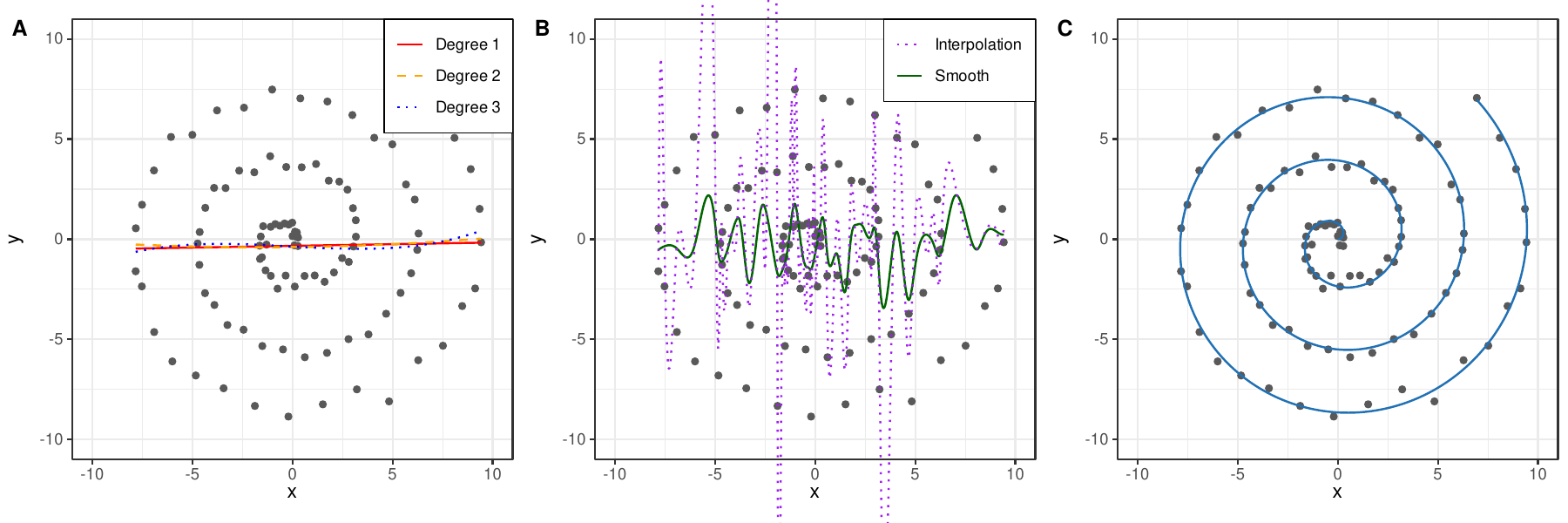}
  \caption{Different modeling approaches on spiral data. A: polynomial regressions of degrees 1 (red, solid), 2 (orange, dashed), and 3 (blue, dotted). B: smoothing spline (green, solid) and interpolating spline (purple, dotted). C: manifold fit.}
  \label{fig:poly-spline}
\end{figure}

The spiral case suggests a different statistical object. 
Viewed comprehensively across the domains of input $x$ and output $y$, the observations lie near a low-dimensional manifold. 
This intuition can be formalized through the latent parametrization of the spiral. 
A continuous latent variable $t$ indexes position along the curve, allowing both the input and output to be modeled as functions of that intrinsic latent state,
\begin{equation}
\label{eq:spiral}
\begin{cases}
x=\alpha t\cos(t),\\
y=\alpha t\sin(t)+\epsilon_y,
\end{cases}
\end{equation}
where $\epsilon_y\sim\mathcal N(0,\sigma^2)$ is output noise and $\alpha>0$ is an unknown scale parameter. 
Under this manifold perspective, the input and output variables are decoupled from each other and instead conditionally depend on the common underlying state $t$. 

This generative structure clarifies the mechanism behind the noto pattern.
Distinct latent roots $t_a\neq t_b$ can produce the exact same input projection $\alpha t_a \cos(t_a)=\alpha t_b \cos(t_b)$ with distinct outputs $\alpha t_a \sin(t_a)\neq \alpha t_b \sin(t_b)$. 
In this view, prediction for each input should be the full set of admissible outputs induced by the manifold.

The immediate advantage of adopting this latent representation is that recovering $t$ converts the problem into two ordinary curve fitting tasks. 
For the spiral, the radial relation \(r_i=\sqrt{x_i^2+y_i^2}\approx \alpha t_i\) permits a specialized phase decoupling strategy. 
Sorting observations by their radius and sequentially unwrapping the corresponding phase yields $\hat t_i=\arctan(y_i/x_i)$. If $\hat t_{i+1}<\hat t_i$, set $\hat t_{i+1}=\hat t_i+2\pi$.
Finally, $\hat\alpha_{LS}$ is estimated via least squares from the nearly linear relation between $r_i$ and $\hat t_i$, namely $\hat\alpha_{LS}=\sum_i r_i \hat{t}_i/\sum_i \hat{t}_i^2$.
When mapped back to the input and output domains, the recovered curves smoothly trace the underlying spiral manifold, as shown in Figure~\ref{fig:poly-spline}C.
Figure~\ref{fig:latent-curves} shows the corresponding latent-space representation. 
Solid and dashed lines in black, green, blue, and red depict the true and estimated coordinate functions; gray points denote observations. The plot demonstrates that the tangled spiral is expressed through two smooth coordinate functions once the latent variable and manifold structure are recovered.

\begin{figure}[tb]
  \centering
  \includegraphics[width=1\textwidth]{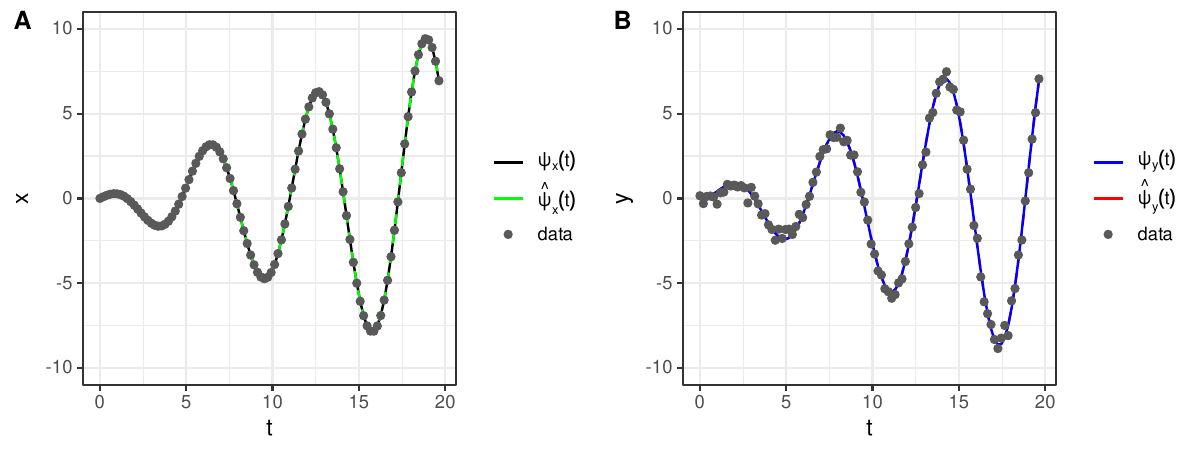}
  \caption{Latent coordinate functions for the spiral. A: $x=\alpha t\cos(t)$ (black, solid) and estimate $\hat{x}=\hat{\alpha}\hat{t}\cos(\hat{t})$ (green, dashed). B: $y=\alpha t\sin(t)$ (blue, solid) and estimate $\hat{y}=\hat{\alpha}\hat{t}\sin(\hat{t})$ (red, dashed). Gray points denote observations.}
  \label{fig:latent-curves}
\end{figure}

The general lesson from this exercise extends well beyond the spiral: when the observed relation is noto, modeling can proceed through latent geometric structure rather than direct regression in the observed space. 
The phase decoupling solver above is specific to polar coordinates. 
While the example is valuable for intuition, it is not sufficient for general noto systems. 
This motivates the need for a systematic manifold modeling approach, developed in Section~\ref{sec:modeling}.

\section{Manifold Regression Model}
\label{sec:modeling}


The manifold modeling objective is to estimate the manifold structure from observations. 
This section first formulates the objective as a regularized profile optimization problem and then presents a robust solver for the objective.

\subsection{Regularized Profile Optimization}
\label{subsec:rpo}

Let $\mathcal{X} \subset \mathbb{R}$ denote the input space and $\mathcal{Y} \subset \mathbb{R}$ denote the output space.
The denoised data are assumed to lie on a one-dimensional manifold whose latent coordinates map respectively into $\mathcal{X}$ and $\mathcal{Y}$.
Specifically, the system is governed by a latent coordinate $t\in \mathbb R$ and shape parameters $\bm\theta$ through a smooth generative map $\mathbf{\Psi}(t;\bm\theta)=[\psi_x(t;\bm\theta)^\top,\psi_y(t;\bm\theta)^\top]^\top$ to the spaces $\mathcal{X}$ and $\mathcal{Y}$.
Unless otherwise stated, the observations are assumed to be independent and identically distributed.
The population sampling model is
$
X_i=\psi_x(T_i;\bm\theta^\star),~
Y_i=\psi_y(T_i;\bm\theta^\star)+\varepsilon_{y,i},
$
where $(T_i,\varepsilon_{y,i})$ are i.i.d., $T_i$ has support $[0,\eta^\star]$ for some $0<\eta^\star<\infty$, and $\varepsilon_{y,i}\sim\mathcal{N}(0,\sigma^2)$ is independent of $T_i$.
Here $\sigma^2>0$ is unknown and $\bm\theta^\star\in\Theta$ denotes the true shape parameters.
The observed pairs $(x_i,y_i)$ are realizations of $(X_i,Y_i)$.
The input coordinate is observed without measurement error, while only the output component is corrupted by additive noise.
Moreover, $\mathbf{\Psi}(\cdot;\bm\theta)$ is assumed to be injective on the relevant latent domain, so that each point on the noiseless manifold has a unique latent coordinate.

Each observed input $x_i$ defines the constraint $\psi_x(t_i;\bm\theta)=x_i$ on the latent variable $t_i$. 
However, without domain restriction, the set
$\{t_{i,j}\in\mathbb{R}:~\psi_x(t_{i,j};\bm\theta)=x_i\}$
can contain infinitely many points, as in the spiral case, so the search cannot be completed in finite time.
Thus, the active domain volume $\eta$ is introduced to truncate the latent domain to $[0,\eta]$.
Let $\mathcal{T}_{\bm\theta,\eta}(x_i)=\{t_{i,j}\in[0,\eta]:\psi_x(t_{i,j};\bm\theta)=x_i\}$
denote the feasible set of latent coordinates corresponding to input $x_i$ restricted to the latent domain $[0,\eta]$.
If $\mathcal{T}_{\bm\theta,\eta}(x_i)\neq\varnothing$, let $\hat{t}_i(\bm\theta,\eta)=\arg\min_{t_{i,j}\in\mathcal{T}_{\bm\theta,\eta}(x_i)}\lVert y_i-\psi_y(t_{i,j};\bm\theta)\rVert^2$ denote the selected latent estimate that best models the observed output $y_i$ from $\mathcal{T}_{\bm\theta,\eta}(x_i)$, and
$
m_{\bm\theta,\eta}(x,y)
=
\inf_{t\in\mathcal{T}_{\bm\theta,\eta}(x)}
\|y-\psi_y(t;\bm\theta)\|^2
$
denotes the corresponding modeling error for a single observation $(x,y)$ given $\bm\theta$ and $\eta$.
If $\mathcal{T}_{\bm\theta,\eta}(x)=\varnothing$, the convention $m_{\bm\theta,\eta}(x,y)=+\infty$ is adopted.

The sample residual error for $\{(x_i,y_i)\}_{i=1}^n$ is then defined as the average modeling error over all observations,
\begin{equation}
  \label{eq:profile-loss}
  R_n(\bm\theta,\eta)={1}/{n}\sum_{i=1}^n m_{\bm\theta,\eta}(x_i,y_i).
\end{equation}
Since the latent coordinates are profiled out through the selected estimates $\{\hat{t}_i(\bm\theta,\eta)\}_{i=1}^n$, the objective $R_n(\bm\theta,\eta)$ depends only on the shape parameters $\bm\theta$ and the domain volume $\eta$.
Motivated by the profile likelihood idea of eliminating nuisance parameters before optimizing shape parameters \citep{murphy2000profile}, the criterion $\min_{\bm\theta,\eta} R_n(\bm\theta,\eta)$ is referred to as \emph{profile optimization} (PO).

The profile optimization can drift toward overly complex manifolds, however.
Figure~\ref{fig:alpha-grid} shows three representative solutions under different values of the domain volume $\eta$. The panels display the fitted manifold curves under different domain volumes. 
Large $\eta$ in Figure~\ref{fig:alpha-grid}A yields many dense spiral turns; 
while small $\eta$ in Figure~\ref{fig:alpha-grid}C yields too few turns and fails to cover the observed structure. 
The middle case in Figure~\ref{fig:alpha-grid}B provides a more plausible balance. 
The behavior in Figure~\ref{fig:alpha-grid}A is characterized as \emph{manifold overfitting}: by inflating $\eta$, the model wraps around the plane repeatedly to reduce residual error, but introduces spurious revolutions unsupported by the data-generating mechanism.

\begin{figure}[tb]
  \centering
  \includegraphics[width=1\textwidth]{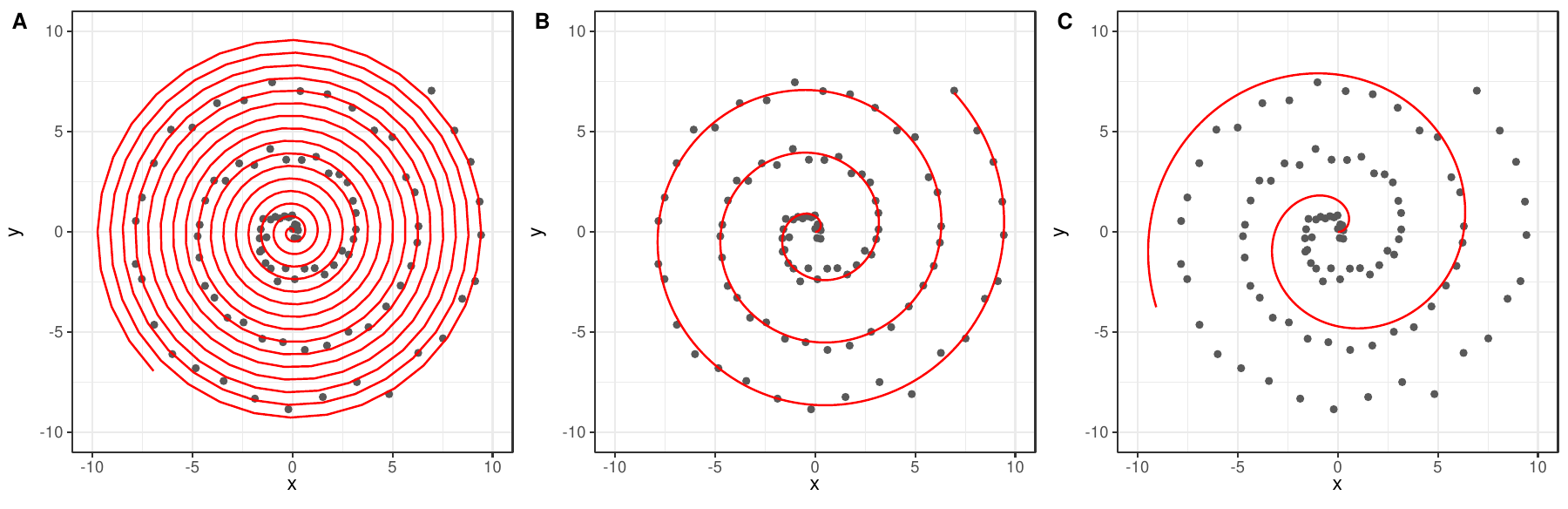}
  \caption{Effect of domain volume $\eta$ under PO on spiral data. A: large $\eta$ ($\eta\approx98$, $\alpha=0.1$) produces dense spurious revolutions. B: intermediate $\eta$ ($\eta\approx19.6$, $\alpha=0.5$) matches the true spiral. C: small $\eta$ ($\eta\approx9.8$, $\alpha=1.0$) fails to cover the data. 
  Red solid curves show fitted manifolds; gray points show observations.}
  \label{fig:alpha-grid}
\end{figure}

\begin{proposition}
\label{prop:monotonicity}
(i) For any given $\bm\theta_0\in\Theta$ and any $\eta_1\le \eta_2$, $R_n(\bm\theta_0,\eta_2)\le R_n(\bm\theta_0,\eta_1)$. 
(ii) For any $\eta_1\le \eta_2$, $\inf_{\bm\theta\in\Theta} R_n(\bm\theta,\eta_2) \le \inf_{\bm\theta\in\Theta} R_n(\bm\theta,\eta_1)$.
(iii) If the optimization problem \\
\(\min_{\bm\theta\in\Theta,\ \eta\in[\eta_{\min},\eta_{\max}]}R_n(\bm\theta,\eta)\), 
with \(\eta_{\max}<\infty\), 
admits a global minimizer, then at least one global minimizer 
\((\widetilde{\bm\theta},\widetilde\eta)\) satisfies \(\widetilde\eta=\eta_{\max}\).

\end{proposition}

Proposition~\ref{prop:monotonicity} formalizes this phenomenon: enlarging the latent domain expands the feasible search region for each observation, so $R_n(\bm\theta,\eta)$ is non-increasing in $\eta$.
Further, on any bounded interval, a globally optimal solution of PO can always be found at the upper boundary. 
This boundary degeneracy of PO motivates an additive regularization term on the active domain volume $\eta$ to prevent manifold overfitting.

To make such a penalty comparable across noise scales, the data-fitting component should be invariant to the unknown output noise variance.
Recall that the output noise is assumed to be $\varepsilon_{y,i}\sim\mathcal{N}(0,\sigma^2)$. 
After profiling out the latent coordinates, the Gaussian log-likelihood is
\begin{align}
  \label{eq:profile-loglike}
\ell(\bm\theta,\eta,\sigma^2)
=-{n}/{2}\log(2\pi\sigma^2)-{1}/{2\sigma^2}\sum_{i=1}^n\left\{y_i-\psi_y(\hat{t}_i(\bm\theta,\eta);\bm\theta)\right\}^2.
\end{align}
For fixed $(\bm\theta,\eta)$, maximizing with respect to $\sigma^2$ gives
$
\hat{\sigma}^2(\bm\theta,\eta)
={1}/{n}\sum_{i=1}^n\{y_i-\psi_y(\hat{t}_i(\bm\theta,\eta);\bm\theta)\}^2.
$
Substituting this estimate back into the likelihood~\eqref{eq:profile-loglike} and neglecting constants independent of $(\bm\theta,\eta)$ yields the negative profiled log-likelihood contribution $n/2\log\left\{R_n(\bm\theta,\eta)\right\}$, where $R_n(\bm\theta,\eta)$ is exactly the profile sample residual error in Equation~\eqref{eq:profile-loss}.
Adding the domain-volume penalty and dividing by the positive constant $n$ gives the normalized regularized profile optimization objective (RPO)
\begin{equation}
  \label{eq:rpo-objective}
J_n(\bm\theta,\eta)
=
\frac{1}{n}\left\{\frac{n}{2}\log R_n(\bm\theta,\eta)+\lambda\eta\right\}
=
\frac12\log R_n(\bm\theta,\eta)
+
\frac{\lambda}{n}\eta,
\end{equation}
where $\lambda>0$ is a fixed regularization parameter; we use $\lambda=1$ throughout the numerical studies.
Note that the first term in \eqref{eq:rpo-objective} corresponds to the data-fitting component and is scale-invariant with respect to the noise variance $\sigma^2$, while the second term penalizes large domain volumes.

The penalty is placed on \(\eta\) because a larger domain volume expands the root search region and may collect remote latent matches that artificially reduce the empirical residual. 
Thus, \(\eta\) should be interpreted as the active domain volume used by the finite-sample fit. 
The manifold shape is represented by \(\bm\theta\), while \(\eta\) only controls the active portion of the fitted manifold used by the finite-sample fit and by downstream slices.

This subsection concludes with asymptotic consistency of the RPO estimator.
Let $R(\bm\theta,\eta)=\mathbb E\{m_{\bm\theta,\eta}(X,Y)\}$ and $J(\bm\theta,\eta)={1}/{2}\log R(\bm\theta,\eta)$.
The feasible latent set \(\mathcal T_{\bm\theta,\eta}(x)\) may be empty for some parameter values, and its root structure may change discontinuously as \((\bm\theta,\eta)\) varies. 
The consistency analysis is therefore stated on an admissible search set. 
A set \(\mathcal A\subseteq\Theta\times[\eta_{\min},\eta_{\max}]\) is called admissible if
$
\mathbb P\bigl\{\mathcal T_{\bm\theta,\eta}(X)\neq\varnothing
\text{ for every }(\bm\theta,\eta)\in\mathcal A\bigr\}=1
$,
and \(m_{\bm\theta,\eta}(x,y)\) is measurable in \((x,y)\) and continuous in \((\bm\theta,\eta)\) on \(\mathcal A\) for almost every \((x,y)\).
Throughout the analysis below, the RPO estimator is defined over an admissible search set \(\mathcal A\).

\begin{lemma}
\label{lem:rpo-uniform-convergence}
Under the i.i.d.\ sampling mechanism, suppose that \(\mathcal A\) is a nonempty compact admissible set and \(\psi_y(t;\bm\theta)\) is continuous on \([0,\eta_{\max}]\times\Theta_{\mathcal A}\), where \(\Theta_{\mathcal A}\) denotes the projection of \(\mathcal A\) onto \(\Theta\), then
$
\sup_{(\bm\theta,\eta)\in\mathcal A}
\bigl|R_n(\bm\theta,\eta)-R(\bm\theta,\eta)\bigr|
\to0
$
almost surely.
\end{lemma}

\begin{lemma}
\label{lem:rpo-existence}
Under the conditions of Lemma~\ref{lem:rpo-uniform-convergence}, suppose additionally that
$
\inf_{(\bm\theta,\eta)\in\mathcal A}R(\bm\theta,\eta)>0,
$
then  
$
\mathbb P\{
\operatorname*{arg\,min}_{(\bm\theta,\eta)\in\mathcal A}
J_n(\bm\theta,\eta)\neq\varnothing
\}\rightarrow 1.
$
\end{lemma}

Let \((\hat{\bm\theta}_n,\hat\eta_n)\) be a measurable estimator
taking values in \(\mathcal A\) and satisfying
$
\mathbb P\{
(\hat{\bm\theta}_n,\hat\eta_n)
\in
\operatorname*{arg\,min}_{(\bm\theta,\eta)\in\mathcal A} \allowbreak
J_n(\bm\theta,\eta)
\}
\to1.
$
To state identification after eliminating the truncation variable, define
$
J_p(\bm\theta)
=
\inf\{J(\bm\theta,\eta):(\bm\theta,\eta)\in\mathcal A\},~
\bm\theta\in\Theta_{\mathcal A}.
$

\begin{theorem}
\label{thm:rpo-consistency}
Under the conditions of Lemmas~\ref{lem:rpo-uniform-convergence} and~\ref{lem:rpo-existence}, suppose that there exists $\bm\theta^\star\in\Theta_{\mathcal A}$ such that, for any $\epsilon>0$,
$
\inf_{\substack{\bm\theta\in\Theta_{\mathcal A}:\|\bm\theta-\bm\theta^\star\|\ge\epsilon}}
J_p(\bm\theta)
>
J_p(\bm\theta^\star),
$
then $\hat{\bm\theta}_n\xrightarrow{p}\bm\theta^\star$.
\end{theorem}

Theorem~\ref{thm:rpo-consistency} applies a standard M-estimation consistency argument \citep{vandervaart1998asymptotic} to the profiled criterion on a nonempty compact admissible set.
Lemmas~\ref{lem:rpo-uniform-convergence} and~\ref{lem:rpo-existence} supply the empirical uniform convergence and minimizer existence used by the argument, while the identification condition in Theorem~\ref{thm:rpo-consistency} identifies the profiled population minimizer.
The theory remains conditional on the existence of such an admissible set; computation over the full restricted box is treated in the next subsection.

\subsection{A Soft-Constraint Solver for RPO}
\label{subsec:solver}

Theorem~\ref{thm:rpo-consistency} establishes consistency of the exact RPO estimator over the admissible set \(\mathcal A\subseteq\Theta\times[\eta_{\min},\eta_{\max}]\). 
In computation, however, \(\mathcal A\) is unknown. 
The algorithm can only search over the restricted box \(\mathcal B=\Theta\times[\eta_{\min},\eta_{\max}]\).
This creates a root set instability problem. 
For a fixed \((\bm\theta,\eta)\), RPO requires solving \(\psi_x(t_i;\bm\theta)=x_i\) exactly and enumerating the feasible latent root set \(\mathcal T_{\bm\theta,\eta}(x_i)\). 
When this root set is oscillatory, small perturbations of \((\bm\theta,\eta)\) may cause roots to appear, merge, disappear, or cross the boundary of \([0,\eta]\). As a result, the selected profiled latent coordinate \(\hat t_i(\bm\theta,\eta)\) may jump discontinuously as the parameter varies.


To avoid the fragility caused by exact root finding, a soft-constrained augmented-Lagrangian solver for RPO is proposed. 
The augmented Lagrangian method is a standard approach for equality-constrained optimization \citep{Bertsekas1999,NocedalWright2006}. 
For an equality constraint \(r=0\), it combines a Lagrange multiplier term \(u\), which records the local direction of constraint violation, with a quadratic penalty term \(\rho\), which discourages large residuals. 
In our setting, the equality constraint is the input matching \(\psi_x(t;\bm\theta)=x\). 
For a single observation \((x,y)\), let
\(r(t;\bm\theta,x)=\psi_x(t;\bm\theta)-x\) denote the input residual, and define
\[
\widetilde m_{\bm\theta,\eta}(x,y;u,\rho)
=
\inf_{t\in[0,\eta]}
\left\{
\|y-\psi_y(t;\bm\theta)\|^2
+
u\,r(t;\bm\theta,x)
+
{\rho}/{2}r^2(t;\bm\theta,x)
\right\}.
\]
Here \(u\) and \(\rho>0\) are fixed multiplier and penalty values. 
Its shifted quadratic form can be rewritten as
$
\widetilde m^{+}_{\bm\theta,\eta}(x,y;u,\rho)
=
\inf_{t\in[0,\eta]}
[
\{y-\psi_y(t;\bm\theta)\}^2
+
u r(t;\bm\theta,x)
+
{\rho}/{2}r^2(t;\bm\theta,x)
+
{u^2}/{2\rho}
].
$
Since \({u^2}/{2\rho}\) is independent of \(t\), the two criteria have the same minimizer, but the shifted version is always nonnegative.
Let \(\tilde t(x,y;u,\rho)\) be the corresponding soft-profiled latent coordinate that achieves the infimum.
Thus, for samples $\{(x_i,y_i)\}_{i=1}^n$, fixed multipliers \(\bm u=(u_1,\ldots,u_n)\) and penalty parameter \(\rho>0\), the soft-constrained regularized profile objective (SRPO) with respect to \((\bm\theta,\eta)\) is
\begin{equation}
\label{eq:soft-al-surrogate}
\widetilde J_n(\bm\theta,\eta;\bm u,\rho)
=
\frac12\log \{\widetilde R_n(\bm\theta,\eta;\bm u,\rho) + \delta\}
+
\frac{\lambda}{n}\eta,
~
(\bm\theta,\eta)\in\mathcal B ,
\end{equation}
where
$
\widetilde R_n(\bm\theta,\eta;\bm u,\rho)
=
{1}/{n}
\sum_{i=1}^n
\widetilde m^+_{\bm\theta,\eta}(x_i,y_i;u_i,\rho)
$
is the soft-constrained residual error, and \(\delta=10^{-10}\) is a small positive constant added to ensure numerical stability of the logarithm.
Moreover, because the residual terms only require the selected latent coordinates to lie in \([0,\eta]\) and the penalty \((\lambda/n)\eta\) increases as \(\eta\) increases, the fitted domain volume can be represented by its endpoint:
$
\hat\eta_n=\max\{\eta_{\min},\max_{1\le i\le n}\tilde t_i\}.
$
Further, when \(\eta_{\min}=0\), \(\hat\eta_n=\max_i\tilde t_i\).

Proposition~\ref{prop:al-continuity} below states the regularity of the soft-constrained surrogate, which explains why SRPO is numerically better behaved than the exact RPO target. 
The \(J_n(\bm\theta,\eta)\) may fail to be continuous outside the unknown admissible set \(\mathcal A\), because the feasible root set itself may change discontinuously. 
By contrast, the soft-constrained \(\widetilde J_n\) is continuous on the full restricted box \(\mathcal B\) for fixed \(\rho\) and bounded multipliers. 
This gives a stable objective for the search of \((\bm\theta,\eta)\).
\begin{proposition}
\label{prop:al-continuity}
Suppose that \(\Theta\) is compact, \(0\le\eta_{\min}<\eta_{\max}<\infty\), and both \(\psi_x(t;\bm\theta)\) and \(\psi_y(t;\bm\theta)\) are continuous in \((t,\bm\theta)\) on \([0,\eta_{\max}]\times\Theta\).
Then, for any fixed samples \(\{(x_i,y_i)\}_{i=1}^n\) and fixed \(0<\rho<\infty\), \(\widetilde J_n(\bm\theta,\eta;\bm u,\rho)\) is continuous in \((\bm\theta,\eta,\bm u)\) on $\Theta\times[\eta_{\min},\eta_{\max}]\times[-U,U]^n$ for any finite \(U>0\).
\end{proposition}

The complete algorithm for SRPO is summarized in Algorithm~\ref{alg:al-rpo}. 
At each iteration, the multipliers and penalty parameter $(\bm u^{(k)},\rho^{(k)})$ are fixed, the shape parameters and the domain volume are updated by minimizing \(\widetilde J_n\), and the resulting soft-profiled latent coordinates are used to compute the input residuals. 
These residuals update the multipliers, while the penalty parameter is increased only when residual contraction is insufficient. 
In this way, SRPO performs manifold modeling through soft latent matching and gradually moves the soft latent matches toward exact feasible roots, while avoiding fragile exact root enumeration during the parameter search.

\begin{algorithm}[ht]
\caption{Soft-constrained solver for RPO (SRPO)}
\label{alg:al-rpo}

\textbf{Input:} \(\{(x_i,y_i)\}_{i=1}^n\), regularization parameter \(\lambda\), restricted box \(\mathcal B=\Theta\times[\eta_{\min},\eta_{\max}]\).

\begin{algorithmic}[1]
\State Initialize \((\bm\theta^{(0)},\eta^{(0)})\in\mathcal B\), multipliers \(\bm u^{(0)}=\bm 0\), penalty parameter \(\rho^{(0)}=1\), contraction threshold \(\tau=1/2\), and penalty inflation factor \(\kappa_\rho=2\).

\For{\(k=0,1,2,\ldots\) until convergence}

  \State \textbf{Parameter update:} \((\bm\theta^{(k+1)},\eta^{(k+1)}) \in \arg\min_{(\bm\theta,\eta)\in\mathcal B} \widetilde J_n(\bm\theta,\eta;\bm u^{(k)},\rho^{(k)})\).

  \State \textbf{Latent coordinate update:} 
  $
  \widetilde t_i^{(k+1)}\in
  \arg\min_{t\in[0,\eta^{(k+1)}]}
  [
    \{y_i-\psi_y(t;\bm\theta^{(k+1)})\}^2
    +u_i^{(k)} r(\allowbreak t;\bm\theta^{(k+1)},x_i)
    +\frac{\rho^{(k)}}{2}r^2(t;\bm\theta^{(k+1)},x_i)
  ], 
  $
  for \(i=1,\ldots,n\).

  \State Compute \(r_i^{(k+1)} = \psi_x(\widetilde t_i^{(k+1)};\bm\theta^{(k+1)}) - x_i\) for \(i=1,\ldots,n\).

  \State \textbf{Multiplier update:} \(u_i^{(k+1)} = u_i^{(k)} + \rho^{(k)}r_i^{(k+1)}\), for \(i=1,\ldots,n\).

  \State \textbf{Penalty update:} Let \(\Delta^{(k+1)}=\max_{1\le i\le n}\|r_i^{(k+1)}\|\). If \(k\ge1\) and \(\Delta^{(k+1)}>\tau\Delta^{(k)}\), set \(\rho^{(k+1)}=\kappa_\rho\rho^{(k)}\); otherwise set \(\rho^{(k+1)}=\rho^{(k)}\).

  \State Stop if \(\Delta^{(k+1)}<\mathrm{tol}_r\).

\EndFor

\State Let \(K\) be the final iteration. 
Set \(\hat{\bm\theta}_n=\bm\theta^{(K)}\), \(\widetilde t_i=\widetilde t_i^{(K)}\) for \(i=1,\ldots,n\), and \(\hat\eta_n=\max\{\eta_{\min},\max_i\tilde t_i\}\).
\end{algorithmic}

\textbf{Output:} Estimated \((\hat{\bm\theta}_n,\hat\eta_n)\), and \(\{\widetilde t_i\}_{i=1}^n\).
\end{algorithm}

The default initial values \(\bm u^{(0)}=\bm 0\), \(\rho^{(0)}=1\), \(\tau=1/2\), and \(\kappa_\rho=2\) follow standard augmented-Lagrangian practice. Zero multipliers correspond to a quadratic-penalty start when no multiplier information is available. The choice \(\rho^{(0)}=1\) keeps the augmented term on the same order as the data-fitting term at the outset. The threshold \(\tau=1/2\) requires the maximum input residual to be reduced by at least one half before the penalty is held fixed, and \(\kappa_\rho=2\) doubles the penalty whenever this contraction is insufficient.

The next proposition shows that the soft-constrained latent update locally recovers an exact feasible root as the penalty parameter increases. 
Recall that \(\hat t_i(\bm\theta,\eta)\in (0,\eta)\) is the selected exact feasible latent root satisfying \(\psi_x(\hat t_i(\bm\theta,\eta);\bm\theta)=x_i\).

\begin{proposition}
\label{prop:al-root-approx}
Suppose that \(\psi_x(\cdot;\bm\theta)\) and \(\psi_y(\cdot;\bm\theta)\) are twice continuously differentiable in a neighborhood of \(\hat t_i(\bm\theta,\eta)\), and that \(\partial_t\psi_x(\hat t_i(\bm\theta,\allowbreak\eta); \allowbreak \bm\theta)\neq0\).
Then there exists \(\delta>0\), with \(I_\delta=[\hat t_i(\bm\theta,\eta)-\delta,~\hat t_i(\bm\theta,\eta)+\delta]\subset[0,\eta]\), such that for every bounded \(u_i\) and all sufficiently large \(\rho\), any minimizer
\[
\widetilde t_i(\bm\theta,\eta;u_i,\rho)
\in
\arg\min_{t\in I_\delta}
\left\{
\|y_i-\psi_y(t;\bm\theta)\|^2
+
u_i r(t;\bm\theta,x_i)
+
{\rho}/{2}r^2(t;\bm\theta,x_i)
\right\}
\]
satisfies \(|\widetilde t_i(\bm\theta,\eta;u_i,\rho)-\hat t_i(\bm\theta,\eta)| = O(\rho^{-1})\).
\end{proposition}

Proposition~\ref{prop:al-root-approx} is a local statement around a stable exact root. 
The differentiability assumption permits a Taylor expansion, and the simple root condition \(\partial_t\psi_x(\hat t_i(\bm\theta,\eta);\bm\theta)\neq0\) is the one-dimensional constraint qualification ensuring that the input equation is locally well behaved. 
Under these conditions, the soft-profiled latent coordinate approaches the corresponding exact latent root as the penalty increases. 
Thus the exact RPO objective \(J_n(\bm\theta,\eta)\) remains the statistical target, while the soft-constrained surrogate \(\widetilde J_n(\bm\theta,\eta;\bm u,\rho)\) provides a smoother and more stable computational path on the full restricted box \(\mathcal B\): for finite \(\rho\), it replaces exact feasibility with soft latent matching, and as \(\rho\) increases with multiplier updates, these soft matches are gradually pushed toward exact input feasibility.

\begin{figure}[t]
  \centering
  \includegraphics[width=0.4\textwidth]{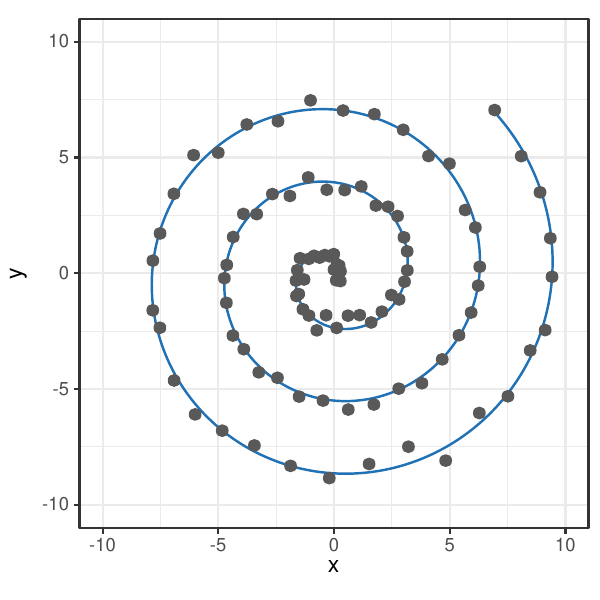}
  \caption{Archimedean spiral data ($n=100$, output noise SD $=0.3$) with SRPO manifold fit (blue solid curve) overlaid on the scatter plot.}
  \label{fig:srpo-spiral}
\end{figure}

\begin{example}[Spiral modeling]
\label{ex:spiral-modeling}
Consider the Archimedean spiral model in \eqref{eq:spiral} with \(\alpha^\star=0.5\), \(\eta^\star=6\pi+\pi/4\), \(n=100\), and \(\epsilon_y\sim\mathcal N(0,0.3^2)\). The soft-constrained solver uses \(\lambda=1\), at most \(30\) outer augmented-Lagrangian iterations, and at most \(100\) iterations per parameter update. 
The fitted parameters are \(\hat\alpha=0.50\) and \(\hat\eta=19.64\). Figure~\ref{fig:srpo-spiral} shows the recovered manifold against the noisy observations.
\end{example}

\section{Manifold Prediction}
\label{sec:prediction}

Unlike the one-to-one cases, prediction is non-trivial for noto cases.
This section proposes a slicing mechanism to utilize the estimated manifold to produce set-valued predictions. Specifically, for a given input \(x_0\), manifold sliced prediction intersects the estimated manifold with the vertical slice \(x=x_0\) and returns all corresponding output coordinates. 
Similar slicing can be performed for any target output value \(y_0\) by intersecting with the horizontal slice \(y=y_0\).
The prediction set is naturally set-valued because a slice can intersect the same manifold at multiple latent positions.

\subsection{Manifold Sliced Prediction}

For a candidate parameter pair \((\bm\theta,\eta)\), 
manifold sliced prediction at input \(x\) is defined by
\[
\mathcal F_{\bm\theta,\eta}(x)
=
\left\{
\psi_y(t;\bm\theta):\ t\in[0,\eta],\ \psi_x(t;\bm\theta)=x
\right\}.
\]
Here \(\eta\) plays the same role as in manifold modeling: it determines the active latent domain over which the slice is taken, so remote roots outside \([0,\eta]\) are excluded.

Computationally, for a fixed query \(x\), we solve \(g_x(t;\bm\theta)=\psi_x(t;\bm\theta)-x=0\) over \(t\in[0,\eta]\).
The implementation uses the R package \texttt{rootSolve} \citep{Soetaert2009RootSolve}.
This routine first subdivides the search interval into several subintervals $(s_{i-1},~s_i)$, where \(0=s_0<s_1<\cdots<s_m=\eta\).
If \(g_x(s_{\ell-1};\bm\theta)\) and \(g_x(s_\ell;\bm\theta)\) have opposite signs, then \([s_{\ell-1},s_\ell]\) contains a crossing under continuity, and the one-dimensional root solver refines a root inside that subinterval.
Grid points, including the two endpoints \(0\) and \(\eta\), are also retained as candidate roots when \(|g_x(s_\ell;\bm\theta)|\) is below a small tolerance.
All candidate roots are then checked by their residual \(|g_x(t;\bm\theta)|\), nearly duplicated latent roots are merged, and duplicated output values are merged after applying \(\psi_y\).
If the retained roots are \(\widehat t_1(x),\ldots,\widehat t_{m_x}(x)\), the reported sliced prediction is \(\widehat{\mathcal F}_{\bm\theta,\eta}(x)=\{\psi_y(\widehat t_j(x);\bm\theta):j=1,\ldots,m_x\}\), where \(m_x\) is the number of retained roots.
This procedure defines the manifold slicing mechanism.

To define the population target, recall that the latent coordinate \(T\) has support \([0,\eta^\star]\).
At the population level, the ideal slice uses the true shape parameters \(\bm\theta^\star\) and the population domain volume \(\eta^\star\). 
Therefore, the ideal manifold sliced prediction set is \(\mathcal F_{\bm\theta^\star,\eta^\star}(x)\).

Given the estimator $(\hat{\bm\theta}_n,\hat\eta_n)$ obtained from manifold modeling, the corresponding plug-in manifold sliced prediction set is
$
\mathcal F_{\hat{\bm\theta}_n,\hat\eta_n}(x)
=
\{
\psi_y(t;\hat{\bm\theta}_n): t\in[0,\hat\eta_n],  \psi_x(t;\hat{\bm\theta}_n)\allowbreak=x
\}.
$
Thus prediction is a direct downstream consequence of manifold modeling, and it is obtained by solving the input projection equation on the estimated manifold.

The consistency of manifold sliced prediction is established as follows.
The Hausdorff distance \citep{molchanov2005random} is introduced to measure the distance between the plug-in prediction set and the ideal prediction set, which is defined as follows.
\[
d_H(\mathcal C_1,\mathcal C_2)
=
\max\left\{
\sup_{a\in\mathcal C_1}\inf_{b\in\mathcal C_2}\|a-b\|,
\sup_{b\in\mathcal C_2}\inf_{a\in\mathcal C_1}\|a-b\|
\right\},
\]
where \(\mathcal C_1\) and \(\mathcal C_2\) are two non-empty compact subsets.

\begin{theorem}
\label{thm:sliced-prediction}
Suppose that \(\hat{\bm\theta}_n\xrightarrow{p}\bm\theta^\star\), \(\hat\eta_n\xrightarrow{p}\eta^\star\),
and there exist \(\delta>0\) and a neighborhood \(\mathcal N\) of \(\bm\theta^\star\) such that \(\psi_x\) and \(\psi_y\) are continuously differentiable on \([0,\eta^\star+\delta]\times \mathcal N\). 
For a given $x_0$, if the equation \(\psi_x(t;\bm\theta^\star)=x_0\) has finitely many solutions on \([0,\eta^\star+\delta]\): \(0<t_1^\star<\cdots<t_K^\star<\eta^\star\), and \(\partial_t\psi_x(t_k^\star;\bm\theta^\star)\neq0\) for \(k=1,\ldots,K\), 
then
$
d_H\{\mathcal F_{\hat{\bm\theta}_n,\hat\eta_n}(x_0),\mathcal F_{\bm\theta^\star,\eta^\star}(x_0)\}\xrightarrow{p}0.
$
\end{theorem}

Theorem~\ref{thm:sliced-prediction} shows that manifold sliced prediction inherits shape consistency from manifold modeling, but also requires coverage consistency of the active domain volume. 
Specifically, the condition \(\hat\eta_n\xrightarrow{p}\eta^\star\) requires that the domain volume also be recovered, because the plug-in slice is taken over \([0,\hat\eta_n]\), whereas the population slice is taken over \([0,\eta^\star]\).
If \(\hat\eta_n\) is systematically smaller than \(\eta^\star\), some latent states contributing to \(\mathcal F_{\bm\theta^\star,\eta^\star}(x_0)\) may be excluded from the plug-in slice, and the Hausdorff distance cannot generally vanish.

Note that the fitted domain volume is reported in endpoint form, \(\hat\eta_n=\max_i\widetilde t_i\). 
Thus, if the recovered latent coordinates satisfy \(\max_i|\widetilde t_i-T_i|\xrightarrow{p}0\), it is enough to ensure that the sample endpoint \(\max_i T_i\) converges to \(\eta^\star\). 
The next proposition gives a sufficient condition for \(\hat\eta_n\xrightarrow{p}\eta^\star\).

\begin{proposition}
\label{prop:eta-convergence}
If \(T_1,\ldots,T_n\) are independent and identically distributed with support \([0,\eta^\star]\), then \(\max_{1\le i\le n}T_i\xrightarrow{p}\eta^\star\).
\end{proposition}

Proposition~\ref{prop:eta-convergence} clarifies the role of sampling in estimating the domain volume. 
The endpoint \(\eta^\star\) can only be recovered if the sampled latent states reach arbitrarily close to the right endpoint of the latent support. 
Together with uniform recovery of the latent coordinates, this yields \(\hat\eta_n\xrightarrow{p}\eta^\star\). 
Thus, the domain-volume condition in Theorem~\ref{thm:sliced-prediction} reflects the need for sufficient sample coverage of the latent domain in manifold prediction. Section~\ref{sec:results} revisits this point in the hysteresis case study.

\subsection{Manifold Sliced Inverse Prediction}

The same latent manifold also supports manifold sliced inverse prediction. 
Similarly, for fixed $(\bm\theta,\eta)$ and a specified output value $y$, let
$
\mathcal G_{\bm\theta,\eta}(y)
=
\left\{
\psi_x(t;\bm\theta):\ t\in[0,\eta],\ \psi_y(t;\bm\theta)=y
\right\}.
$
Moreover, the ideal manifold sliced inverse prediction set is denoted by $\mathcal G_{\bm\theta^\star,\eta^\star}(y)$.
Given the estimator $(\hat{\bm\theta}_n,\hat\eta_n)$, the corresponding plug-in manifold sliced inverse prediction set is
$
\mathcal G_{\hat{\bm\theta}_n,\hat\eta_n}(y)
=
\left\{
\psi_x(t;\hat{\bm\theta}_n):\ t\in[0,\hat\eta_n],\ \psi_y(t;\hat{\bm\theta}_n)=y
\right\}.
$

Under this formulation, for any fixed \((\bm\theta,\eta)\), \(y\in\mathcal F_{\bm\theta,\eta}(x)\) if and only if \(x\in\mathcal G_{\bm\theta,\eta}(y)\). 
Thus manifold sliced prediction and manifold sliced inverse prediction are two coordinate slices of the same fitted manifold.
A consistency result for manifold sliced inverse prediction follows by symmetry, after exchanging the roles of $\psi_x$ and $\psi_y$.
As a direct corollary of Theorem~\ref{thm:sliced-prediction}, Corollary~\ref{cor:sliced-inverse-prediction} gives a solution for multivalued inverse problems under the estimated manifold structure.

\begin{corollary}
\label{cor:sliced-inverse-prediction}
Suppose that \(\hat{\bm\theta}_n\xrightarrow{p}\bm\theta^\star\), \(\hat\eta_n\xrightarrow{p}\eta^\star\),
and there exist \(\delta>0\) and a neighborhood \(\mathcal N\) of \(\bm\theta^\star\) such that \(\psi_x\) and \(\psi_y\) are continuously differentiable on \([0,\eta^\star+\delta]\times \mathcal N\). 
For a given $y_0$, if the equation \(\psi_y(t;\bm\theta^\star)=y_0\) has finitely many solutions on \([0,\eta^\star+\delta]\): \(0<s_1^\star<\cdots<s_L^\star<\eta^\star\), and \(\partial_t\psi_y(s_\ell^\star;\bm\theta^\star)\neq0\) for \(\ell=1,\ldots,L\),
then
$
d_H\{\mathcal G_{\hat{\bm\theta}_n,\hat\eta_n}(y_0),\mathcal G_{\bm\theta^\star,\eta^\star}(y_0)\}\xrightarrow{p}0.
$
\end{corollary}

The section concludes by revisiting the spiral example and evaluating the forward and inverse prediction errors under the SRPO fit.

\addtocounter{example}{-1}
\renewcommand{\theHexample}{continued.\arabic{example}}
\begin{example}[Continued: Spiral prediction]
\label{ex:spiral-prediction}
Revisit Example~\ref{ex:spiral-modeling} and use its SRPO fit. 
Take $\mathcal X_{\mathrm{test}}$ and $\mathcal Y_{\mathrm{test}}$ to be $1000$-point uniform grids on the x and y ranges. 
The mean forward error 
$\bar d_{\mathcal F}=|\mathcal X_{\mathrm{test}}|^{-1}\sum_{x\in \mathcal X_{\mathrm{test}}} d_H\{\mathcal F_{\hat{\bm\theta}_n,\hat\eta_n}(x),\mathcal F_{\bm\theta^\star,\eta^\star}(x)\}=0.012$ 
and the mean inverse error 
$\bar d_{\mathcal G}=|\mathcal Y_{\mathrm{test}}|^{-1} \allowbreak \sum_{y\in \mathcal Y_{\mathrm{test}}} d_H\{\mathcal G_{\hat{\bm\theta}_n,\hat\eta_n}(y),\mathcal G_{\bm\theta^\star,\eta^\star}(y)\}=0.021$.
\end{example}

The example illustrates that the same estimated manifold induces both manifold sliced prediction and manifold sliced inverse prediction. 
The reported errors are small relative to the scale of the data, confirming that these sliced prediction sets recover the true admissible sets under multivaluedness.

\section{Case Studies}
\label{sec:results}

This section evaluates manifold regression on three representative noto structures beyond the spiral: a hysteresis loop, a prolate trochoid, and a Lissajous curve. 
These examples are chosen to test different geometric features of manifold prediction: two-branch hysteresis, looped open curves, and periodic self-intersections. 
In all examples, data are generated from a noiseless structural manifold with additive noise only in the output coordinate. 
Specifically, the latent states \(t_i\) are placed on an equally spaced grid over the stated latent interval, the inputs are generated as \(x_i=\psi_x(t_i;\bm\theta^\star)\), and the responses are generated as \(y_i=\psi_y(t_i;\bm\theta^\star)+\epsilon_i\), where \(\epsilon_i\sim N(0,0.1^2)\). 
The true set-valued forward prediction and inverse prediction sets are computed from the noiseless structural manifold.

All examples are fitted by the SRPO solver in Algorithm~\ref{alg:al-rpo}. 
The augmented Lagrangian loop is allowed to run for at most \(30\) outer steps. 
Within each outer step, latent profiling is carried out by the L-BFGS-B option in the R function \texttt{optim} \citep{R-optim}, and the shape parameters and the domain volume \((\bm\theta,\eta)\) are updated by the BOBYQA option in the R package \texttt{nloptr} \citep{R-nloptr}, with at most \(100\) BOBYQA iterations. 
The default parameter settings are \(\lambda=1\), \(\rho_0=1\), \(\tau=1/2\), \(\kappa_{\rho}=2\), and a stopping tolerance \(\max_i |r_i|<10^{-3}\).

Prediction accuracy is measured by mean Hausdorff error. 
For set-valued forward prediction, \(d_H\{\mathcal F_{\hat{\bm\theta}_n,\hat\eta_n}(x),\mathcal F_{\bm\theta^\star,\eta^\star}(x)\}\) is averaged over a \(1000\)-point grid on the input range where \(\mathcal F_{\bm\theta^\star,\eta^\star}(x)\) is nonempty. 
The inverse prediction error is defined analogously over the output range where \(\mathcal G_{\bm\theta^\star,\eta^\star}(y)\) is nonempty. 

\subsection{Hysteresis Loop}

The hysteresis case represents a two-branch noto relation in which the same input level may correspond to distinct output states. 
Motivated by classical hysteresis loops \citep{mayergoyz2003mathematical}, the noiseless structural manifold is
\[
\psi_x(t;\bm\theta)=\theta_1\cos t,
~
\psi_y(t;\bm\theta)=\theta_2\tanh\{\cos(t-\theta_3)\}.
\]
The latent states are generated on \(t_i\in[0,2\pi]\), the sample size is \(n=50\), and the true parameters are \(\bm\theta^\star=(2,3,0.5)\) with \(\eta^\star=2\pi\). 
The SRPO search is 
constrained to \(\theta_1\in[0.5,5]\), \(\theta_2\in[1,6]\), \(\theta_3\in[0,2]\), and \(\eta\in[\pi,6\pi]\).

\begin{figure}[tb]
  \centering
  \includegraphics[width=0.8\textwidth]{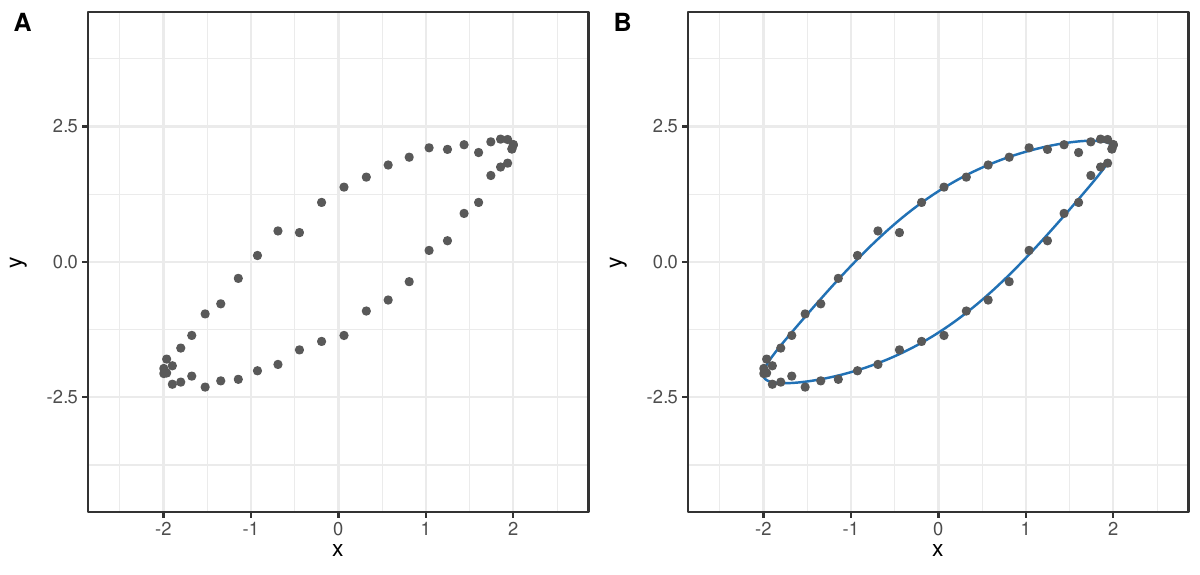}
  \caption{Hysteresis loop data ($n=50$, output noise SD $=0.1$). A: scatter plot of observations. B: SRPO manifold fit (blue solid curve) overlaid on the data.}
  \label{fig:hysteresis_case}
\end{figure}

The SRPO fit in Figure~\ref{fig:hysteresis_case} recovers the two-branch loop from noisy point-valued observations, returning 
\((\hat\theta_1,\hat\theta_2,\hat\theta_3,\hat\eta_n)=(2.002,2.940,0.498,6.023)\). 
That is, \(\hat{\psi}_x(t;\hat{\bm\theta}_n)=2.002\cos t\) and 
\(\hat{\psi}_y(t;\hat{\bm\theta}_n)=2.940\tanh\{\cos(t-0.498)\}\).
The induced set-valued forward and inverse predictions preserve the separated branches, attaining mean Hausdorff errors \(0.045\) and \(0.094\), respectively.

Note that the fitted domain volume \(\hat\eta_n=6.023\) is slightly smaller than the population value 
\(\eta^\star=2\pi\approx6.283\), because the largest sampled latent state \(\max_i t_i\) does not reach the right endpoint of the latent domain.
Consequently, the largest forward Hausdorff error occurs at the input level \(x=1.940\).
At this input level, the true slice has two latent roots, \(t=0.246\) and \(t=6.037\), while the plug-in slice retains only the first root, yielding a Hausdorff distance of \(0.323\) after mapping the roots to the output coordinate.
This finite-sample behavior is consistent with Proposition~\ref{prop:eta-convergence}, which shows that recovery of the population domain volume relies on coverage of the latent domain by the sampled latent states.
This observation also suggests that the experimental design of latent states may affect the efficiency of manifold prediction, although a formal design theory is beyond the scope of the present paper.

\subsection{Prolate Trochoid}

The prolate trochoid case tests an open manifold with repeated branches and looped geometry. 
Trochoids are standard roulette curves, and the prolate case contains loops \citep{yates1952handbook,weisstein2026trochoid}. 
The noiseless structural manifold is
\[
\psi_x(t;\bm\theta)=\theta_1 t-\theta_2\sin t,
~
\psi_y(t;\bm\theta)=\theta_1-\theta_2\cos t.
\]
The latent states are generated on \(t_i\in[0,6\pi]\), the sample size is \(n=50\), and the true parameters are \(\bm\theta^\star=(1,3)\) with \(\eta^\star=6\pi\). 
The SRPO search is 
constrained to \(\theta_1\in[0.2,2]\), \(\theta_2\in[0.3,3]\), and \(\eta\in[2\pi,10\pi]\).

\begin{figure}[tb]
  \centering
  \includegraphics[width=0.8\textwidth]{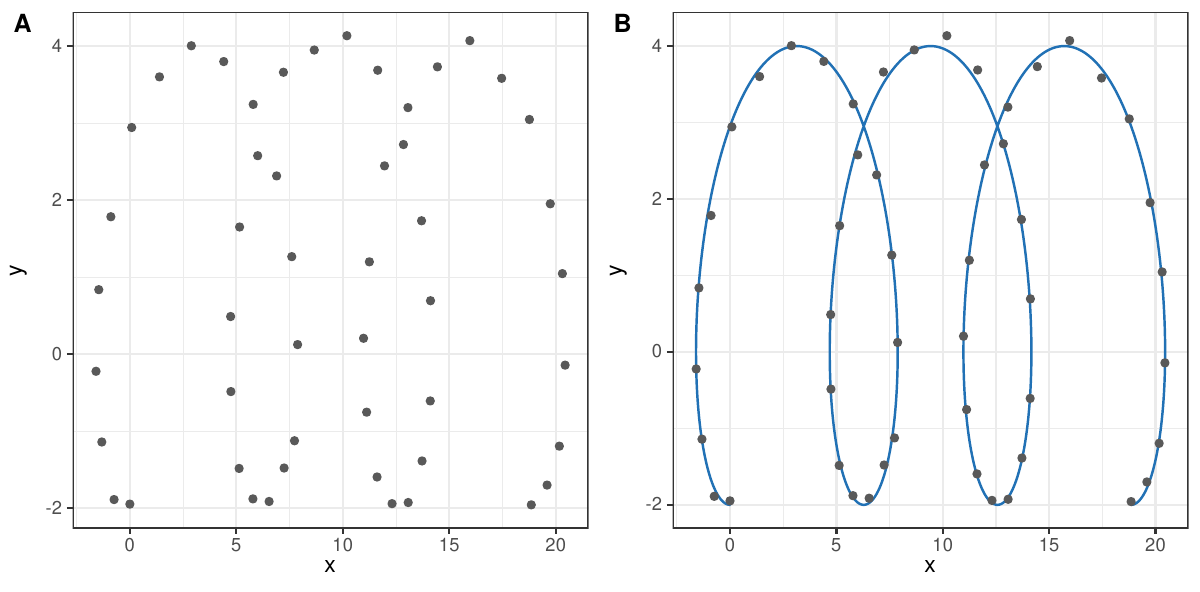}
  \caption{Prolate trochoid data ($n=50$, output noise SD $=0.1$). A: scatter plot of observations. B: SRPO manifold fit (blue solid curve) overlaid on the data.}
  \label{fig:trochoid_case}
\end{figure}

The SRPO fit in Figure~\ref{fig:trochoid_case} accurately represents the repeated branches and looped geometry, returning \((\hat\theta_1,\hat\theta_2,\hat\eta)=(1.000,3.000,18.850)\). 
Namely, $\hat{\psi}_x(t;\hat{\bm\theta}_n)=1.000 t-3.000\sin t$ and $\hat{\psi}_y(t;\hat{\bm\theta}_n)=1.000-3.000\cos t$.
The induced set-valued forward prediction and inverse prediction both attain a mean Hausdorff error \(0.001\).

\subsection{Lissajous Curve}

The Lissajous case tests a periodic multi-branch geometry with repeated overlaps and self-intersections. 
Lissajous curves arise from pairs of orthogonal sinusoidal motions and provide a canonical source of self-intersecting periodic curves \citep{lissajous1857memoire}. 
The noiseless two-parameter structural manifold, with \(\bm\theta=(\theta_1,\theta_2)^\top\), is
\[
\psi_x(t;\bm\theta)=\theta_1\sin(3 t+\pi/2),
~
\psi_y(t;\bm\theta)=\theta_2\sin(4 t).
\]
The latent states are generated on \(t_i\in[0,2\pi]\), the sample size is \(n=100\), and the true two-parameter value is \(\bm\theta^\star=(1,1)\) with \(\eta^\star=2\pi\). 
The SRPO search is 
constrained to \(\theta_1\in[0.1,2]\), \(\theta_2\in[0.1,2]\), and \(\eta\in[\pi,4\pi]\).

\begin{figure}[ht]
  \centering
  \includegraphics[width=0.8\textwidth]{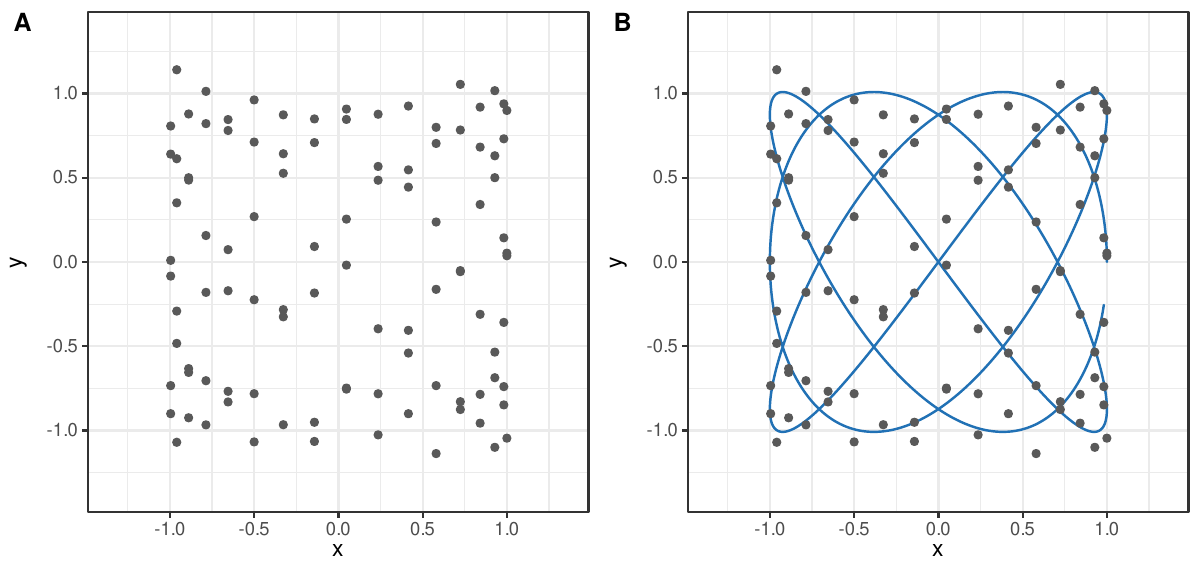}
  \caption{Lissajous curve data ($n=100$, output noise SD $=0.1$). A: scatter plot of observations. B: SRPO manifold fit (blue solid curve) overlaid on the data.}
  \label{fig:lissajous_case}
\end{figure}

The SRPO fit in Figure~\ref{fig:lissajous_case} remains stable under periodic self-intersections, returning \((\hat\theta_1,\hat\theta_2,\hat\eta)=(1.000,1.009,6.219)\). 
Namely, $\hat{\psi}_x(t;\hat{\bm\theta}_n)=1.000\sin(3 t+\pi/2)$ and $\hat{\psi}_y(t;\hat{\bm\theta}_n)=1.009\sin(4 t)$.
The induced set-valued forward prediction and inverse prediction retain multiple admissible branches, attaining mean Hausdorff errors \(0.011\) and \(0.034\), respectively.

Across the three examples, SRPO recovers the latent manifolds from noisy point-valued observations and induces stable set-valued forward prediction and inverse prediction. 
These results show that manifold modeling can recover qualitatively different noto structures and that the induced downstream predictions remain stable.

\section{Conclusion}
\label{sec:conclusion}

This work provides a new statistical formulation of \emph{manifold regression} for non-one-to-one (noto) input-output relations.
In the parametric setting considered here, the manifold and latent mappings have specified functional forms indexed by an unknown finite dimensional parameter.
It treats modeling the latent manifold as the central statistical task and induces set-valued prediction by slicing the learned manifold.
Theoretical guarantees are also stated for manifold modeling and the resulting prediction sets.
It changes the predictive object: instead of forcing a noto relation into a direct regression map, it represents the underlying relation by a latent manifold while preserving the statistical roles of input and output.
This framework also suggests a useful principle: when the target is poorly behaved in the original coordinates, recasting the problem in a suitable latent space may reveal the right statistical object and lead to more principled prediction.

Furthermore, the proposed framework gives a unified view of several familiar predictive tasks. 
When the manifold is $\mathcal{M}=\{(x,f(x)) : x\in\mathcal{X}\}$, manifold prediction reduces to ordinary regression. 
When the object of interest is to find the inputs given a target output, the same estimated manifold induces manifold sliced inverse prediction. 
Thus classical regression and inverse prediction are unified as different projections of the same latent structure.
Consequently, the formulation extends predictive learning from one-to-one relations to noto relations while preserving a coherent framework for modeling and prediction.

Manifold regression also clarifies its relationship with existing paradigms.
Bayesian inverse problems \citep{Stuart2010} usually start from a given forward operator and study recovery of inputs from outputs, whereas the present problem treats the noto input-output relation itself as the object to be learned.
By contrast, the set-valued modeling literature \citep{LiMolchanovMolinariPeng2021} allows set-valued responses, but assumes that the observed response for each input is already a set.
The problem considered here is fundamentally different: observations are standard input-output pairs, while the underlying relation may admit multiple outputs for a given input.
Manifold learning provides related nonparametric tools for geometric representation, while manifold regression focuses on a relation manifold with distinct input and output roles and uses its coordinate slices for prediction.

Several important directions remain primitive. 
One direction is nonparametric manifold regression, where the manifold is estimated directly from the data rather than through a pre-specified parametric family. 
This would link the predictive formulation to principal curves and related manifold modeling methods while preserving the input-output distinction that motivates manifold regression. 
Extensions to higher dimensional settings are also under investigation, including higher dimensional observed spaces and latent manifolds. 
In such settings, the method must accommodate more complex geometric structure, higher dimensional domain volume, and the computational demands of recovering a regression relation from finite samples.
Another direction is uncertainty quantification for manifold regression, including uncertainty in the estimated manifold, the recovered domain volume, and the induced prediction sets. 
Such uncertainty measures would also provide a natural way for studying experimental design efficiency, since new samples could be chosen to reduce uncertainty.

Manifold regression is intended as a unified framework rather than a single-purpose procedure. 
It is hoped that this formulation will motivate further statistical questions and lead to useful tools for modeling, sampling, and prediction in noto input-output problems.

\section{Competing interests}
\label{disclosure-statement}
The authors declare no competing interests.

\section{Data availability statement}
\label{data-availability-statement}
No external data are associated with this article.
The numerical examples and figures can be reproduced using the R code provided in the supplementary materials.

\section{Supplementary material}
\label{supplementary-material}
The supplementary material is organized as follows.

\begin{itemize}
\item \textbf{Spiral data and technical proofs.}
The spiral dataset used in the motivating example and the proofs for all theoretical results are collected in Appendix~\ref{app:supp}.

\item \textbf{Code (.zip).}
Contains R code for implementing the soft-constrained solver for regularized profile optimization (SRPO), manifold prediction, and reproducing the numerical results and figures in the paper.
\end{itemize}

\appendix

\section{Supplementary Material}
\label{app:supp}

This appendix contains material that supports the main manuscript.
It first reports the spiral dataset used in the motivating example,
and then provides technical proofs for all the theoretical results.

\subsection{Spiral Data}

Table~\ref{tab:spiral-data} lists the 100 observations that are displayed in Figure~\ref{fig:spiral_intro} and reused in the spiral modeling and prediction examples in the main manuscript.
The dataset was generated from the Archimedean spiral model \(x=\alpha t\cos(t)\), \(y=\alpha t\sin(t)+\epsilon_y\), \(0\le t\le \eta\), with \(\alpha^\star=0.5\), \(\eta^\star=6\pi+\pi/4\), \(n=100\), and \(\epsilon_y\sim\mathcal N(0,0.3^2)\).

\begin{table}[H]
\centering
\captionsetup{skip=2pt}
\caption{Spiral data used in the motivating example.}
\label{tab:spiral-data}
{\fontsize{10pt}{10pt}\selectfont 
\setlength{\tabcolsep}{2pt}
\renewcommand{\arraystretch}{1} 
\begin{tabular}{@{}rr@{\hspace{1.25em}}rr@{\hspace{1.25em}}rr@{\hspace{1.25em}}rr@{\hspace{1.25em}}rr@{\hspace{1.25em}}rr@{\hspace{1.25em}}rr@{}}
\hline
x & y & x & y & x & y & x & y & x & y & x & y & x & y
\\
\hline
-7.83 & -1.60 & -3.92 & 2.56 & -1.51 & -5.33 & -0.21 & -8.85 & 0.62 & -5.89 & 3.02 & 1.55 & 6.28 & -6.04 \\
-7.83 & 0.55 & -3.90 & -3.28 & -1.47 & 0.65 & -0.19 & 0.73 & 1.11 & -1.81 & 3.05 & -0.36 & 6.30 & 0.28 \\
-7.53 & -2.35 & -3.77 & 6.43 & -1.35 & -1.56 & -0.01 & 0.83 & 1.18 & 3.75 & 3.16 & 0.95 & 6.94 & 7.05 \\
-7.52 & 1.72 & -3.44 & -7.44 & -1.30 & -0.27 & 0.00 & 0.16 & 1.52 & -8.24 & 3.17 & 0.12 & 7.51 & -5.32 \\
-6.92 & 3.43 & -3.34 & 2.55 & -1.11 & 4.14 & 0.10 & -0.30 & 1.60 & -2.13 & 3.22 & -7.50 & 8.09 & 5.06 \\
-6.91 & -4.63 & -3.25 & -4.28 & -1.09 & 0.62 & 0.11 & -2.36 & 1.72 & -5.67 & 3.80 & -4.75 & 8.48 & -3.34 \\
-6.07 & 5.11 & -2.66 & 3.42 & -1.08 & -1.83 & 0.13 & 0.38 & 1.74 & 6.88 & 4.08 & 5.07 & 8.90 & 3.50 \\
-6.00 & -6.10 & -2.45 & -4.52 & -1.02 & 7.47 & 0.18 & 0.12 & 1.81 & 2.92 & 4.68 & -3.72 & 9.12 & -2.46 \\
-5.00 & 5.21 & -2.42 & 6.57 & -0.86 & 0.75 & 0.22 & 0.33 & 2.07 & -1.66 & 4.83 & -8.10 & 9.35 & 1.51 \\
-4.83 & -6.80 & -1.91 & 3.34 & -0.75 & -2.47 & 0.25 & 0.14 & 2.34 & 2.88 & 4.99 & 4.74 & 9.42 & -0.15 \\
-4.74 & -0.21 & -1.88 & -8.32 & -0.63 & 0.67 & 0.27 & -0.34 & 2.48 & -0.94 & 5.40 & -2.68 & & \\
-4.65 & -1.28 & -1.64 & -0.32 & -0.48 & -5.50 & 0.28 & 0.08 & 2.75 & 2.47 & 5.67 & 2.73 & & \\
-4.64 & 0.36 & -1.62 & -0.98 & -0.40 & 0.79 & 0.39 & 7.03 & 2.80 & -4.99 & 5.93 & -1.70 & & \\
-4.37 & -2.69 & -1.59 & 0.13 & -0.34 & -1.82 & 0.46 & 3.59 & 2.81 & -1.13 & 6.11 & 1.98 & & \\
-4.36 & 1.57 & -1.53 & -0.90 & -0.31 & 3.61 & 0.60 & -1.83 & 2.99 & 6.20 & 6.23 & -0.53 & & \\
\hline
\end{tabular}
}
\end{table}

\subsection{Proof of Proposition~\ref{prop:monotonicity}}
\begin{proof}
For each observation \(i\), the inclusion \([0,\eta_1]\subseteq[0,\eta_2]\) gives
\(\mathcal T_{\bm\theta_0,\eta_1}(x_i)\subseteq\mathcal T_{\bm\theta_0,\eta_2}(x_i)\).
Therefore, by the definition of the infimum and the convention that the infimum over an empty feasible set is \(+\infty\),
$
m_{\bm\theta_0,\eta_2}(x_i,y_i)
\le
m_{\bm\theta_0,\eta_1}(x_i,y_i).
$
Averaging over \(i=1,\ldots,n\) gives
\(R_n(\bm\theta_0,\eta_2)\le R_n(\bm\theta_0,\eta_1)\),
which proves (i).

Taking the infimum over \(\bm\theta\in\Theta\) on both sides of the monotonicity in (i) gives
$
\inf_{\bm\theta\in\Theta} R_n(\bm\theta,\eta_2)
\le
\inf_{\bm\theta\in\Theta} R_n(\bm\theta,\eta_1),
$
which proves (ii).

For (iii), suppose that the problem
\(\min_{\bm\theta\in\Theta,\ \eta\in[\eta_{\min},\eta_{\max}]} R_n(\bm\theta,\eta)\)
admits a global minimizer, and let \((\widetilde{\bm\theta},\widetilde\eta)\) be one such minimizer. 
By (i),
\(R_n(\widetilde{\bm\theta},\eta_{\max})\le R_n(\widetilde{\bm\theta},\widetilde\eta)\). Since
\((\widetilde{\bm\theta},\widetilde\eta)\) is already a global minimizer, the reverse inequality also holds:
\(R_n(\widetilde{\bm\theta},\widetilde\eta)\le R_n(\widetilde{\bm\theta},\eta_{\max})\). Hence
\(R_n(\widetilde{\bm\theta},\eta_{\max})=R_n(\widetilde{\bm\theta},\widetilde\eta)\), and therefore
\((\widetilde{\bm\theta},\eta_{\max})\) is also a global minimizer. This proves (iii).
\end{proof}

\subsection{Proof of Lemma~\ref{lem:rpo-uniform-convergence}}

\begin{proof}
By the sampling model, $\operatorname{supp}(T)=[0,\eta^\star]$ with $0<\eta^\star<\infty$.
Continuity of $\psi_y(\cdot;\bm\theta^\star)$ on $[0,\eta^\star]$ yields
$M^\star:=\sup_{t\in[0,\eta^\star]}|\psi_y(t;\bm\theta^\star)|<\infty$.
Independence of $\varepsilon_y$ and $T$, together with $\mathbb E(\varepsilon_y)=0$ and $\mathbb E(\varepsilon_y^2)=\sigma^2$, gives
$
\mathbb E(Y^2)
=
\mathbb E\{\psi_y^2(T;\bm\theta^\star)\}+\sigma^2
\le
(M^\star)^2+\sigma^2
<\infty.
$
Compactness of $\Theta_{\mathcal A}$ and continuity of $\psi_y$ on $[0,\eta_{\max}]\times\Theta_{\mathcal A}$ yield
$M:=\sup_{(t,\bm\theta)\in[0,\eta_{\max}]\times\Theta_{\mathcal A}}|\psi_y(t;\bm\theta)|<\infty$.
By admissibility, for every $(\bm\theta,\eta)\in\mathcal A$,
$
0\le m_{\bm\theta,\eta}(X,Y)
\le
\sup_{0\le t\le\eta_{\max}}\{Y-\psi_y(t;\bm\theta)\}^2
\le
2Y^2+2M^2
$
almost surely, whence $R(\bm\theta,\eta)<\infty$ on $\mathcal A$.

For every $\delta>0$, define
$
\omega_\delta(x,y)
=
\sup\bigl\{
|m_{\bm\theta,\eta}(x,y)-m_{\bm\theta',\eta'}(x,y)|
:
(\bm\theta,\eta),(\bm\theta',\eta')\in\mathcal A,\ 
\|(\bm\theta,\eta)-(\bm\theta',\eta')\|\le\delta
\bigr\}.
$
Continuity on the compact set $\mathcal A$ gives $\omega_\delta(x,y)\to0$ as $\delta\to0$ for almost every $(x,y)$, and
$0\le\omega_\delta(X,Y)\le 2Y^2+2M^2$.
Dominated convergence therefore yields $\mathbb E\{\omega_\delta(X,Y)\}\to0$ as $\delta\to0$.

For every $\delta>0$, compactness of $\mathcal A$ yields a finite $\delta$-net
$\{(\bm\theta_j,\eta_j)\}_{j=1}^{N_\delta}\subset\mathcal A$, that is,
for every $(\bm\theta,\eta)\in\mathcal A$ there exists $j\in\{1,\ldots,N_\delta\}$ such that
$\|(\bm\theta,\eta)-(\bm\theta_j,\eta_j)\|\le\delta$.
Consequently,
$
\sup_{(\bm\theta,\eta)\in\mathcal A}
|R_n(\bm\theta,\eta)-R(\bm\theta,\eta)|
\le
\max_{1\le j\le N_\delta}
|R_n(\bm\theta_j,\eta_j)-R(\bm\theta_j,\eta_j)|
+
n^{-1}\sum_{i=1}^n\omega_\delta(X_i,Y_i)
+
\mathbb E\{\omega_\delta(X,Y)\}.
$
The strong law of large numbers yields
$\max_{1\le j\le N_\delta}|R_n(\bm\theta_j,\eta_j)-R(\bm\theta_j,\eta_j)|\to0$
and
$n^{-1}\sum_{i=1}^n\omega_\delta(X_i,Y_i) \allowbreak \to\mathbb E\{\omega_\delta(X,Y)\}$
almost surely, so
$
\limsup_{n\to\infty}
\sup_{(\bm\theta,\eta)\in\mathcal A}
|R_n(\bm\theta,\eta)-R(\bm\theta,\eta)|
\le
2\mathbb E\{\omega_\delta(X,Y)\}
$
almost surely.
The preceding bound holds for every $\delta>0$; since
$\mathbb E\{\omega_\delta(X,Y)\}\to0$ as $\delta\to0$,
$
\sup_{(\bm\theta,\eta)\in\mathcal A}
|R_n(\bm\theta,\eta)-R(\bm\theta,\eta)|
\to0
$
almost surely.
\end{proof}

\subsection{Proof of Lemma~\ref{lem:rpo-existence}}

\begin{proof}
By the additional lower bound, there exists \(c_0>0\) such that \(R(\bm\theta,\eta)\ge c_0\) for all \((\bm\theta,\eta)\in\mathcal A\).
Lemma~\ref{lem:rpo-uniform-convergence} gives
$
\sup_{(\bm\theta,\eta)\in\mathcal A}
|R_n(\bm\theta,\eta)-R(\bm\theta,\eta)|
\to0
$
almost surely, and hence also in probability.
On the event
$
\sup_{(\bm\theta,\eta)\in\mathcal A}
|R_n(\bm\theta,\eta)-R(\bm\theta,\eta)|
<c_0/2,
$
\(R_n(\bm\theta,\eta)\ge c_0/2\) holds uniformly over \(\mathcal A\).
Thus \(\log R_n\) is finite and continuous on the compact set \(\mathcal A\), and therefore so is
$
J_n(\bm\theta,\eta)
=
\frac12\log R_n(\bm\theta,\eta)
+
\frac{\lambda}{n}\eta.
$
By the Weierstrass theorem, the minimizing set of \(J_n\) over \(\mathcal A\) is nonempty and compact on this event.
Since the event has probability tending to one,
$
\mathbb P\bigl\{\arg\min_{(\bm\theta,\eta)\in\mathcal A}J_n(\bm\theta,\eta)\neq\varnothing\bigr\}\to1.
$
\end{proof}

\subsection{Proof of Theorem~\ref{thm:rpo-consistency}}

\begin{proof}
For each $\bm\theta\in\Theta_{\mathcal A}$, let
\(\mathcal A_{\bm\theta}=\{\eta\in[\eta_{\min},\eta_{\max}]:(\bm\theta,\eta)\in\mathcal A\}\)
denote the admissible slice of latent domain volumes, and define the profiled empirical criterion by
$
J_{n,p}(\bm\theta)
:=
\inf_{\eta\in\mathcal A_{\bm\theta}}
J_n(\bm\theta,\eta).
$
On the event that a global minimizer of $J_n$ over $\mathcal A$ exists,
$
\hat{\bm\theta}_n\in\arg\min_{\bm\theta\in\Theta_{\mathcal A}}
J_{n,p}(\bm\theta).
$

Under the conditions of Lemma~\ref{lem:rpo-existence}, there exists $c_0>0$ such that
$
R(\bm\theta,\eta)\ge c_0
$
for all $(\bm\theta,\eta)\in\mathcal A$. Lemma~\ref{lem:rpo-uniform-convergence} gives
$
\sup_{(\bm\theta,\eta)\in\mathcal A}
|R_n(\bm\theta,\eta)-R(\bm\theta,\eta)|
\to0
$
almost surely, and hence also in probability.
Thus, with probability tending to one,
$
\sup_{(\bm\theta,\eta)\in\mathcal A}
|R_n(\bm\theta,\eta)-R(\bm\theta,\eta)|
<c_0/2,
$
and on this event $R_n(\bm\theta,\eta)\ge c_0/2$ uniformly over $\mathcal A$.
Since $\log(\cdot)$ is Lipschitz on $[c_0/2,\infty)$,
$
\sup_{(\bm\theta,\eta)\in\mathcal A}
\bigl|\frac12\log R_n(\bm\theta,\eta)-\frac12\log R(\bm\theta,\eta)\bigr|
\xrightarrow{p}0.
$
Moreover, because $\lambda>0$ is fixed,
$
\sup_{\mathcal A}(\lambda/n)\eta
\le
\lambda\eta_{\max}/n
\to0.
$
Therefore
$
\sup_{(\bm\theta,\eta)\in\mathcal A}
\bigl|J_n(\bm\theta,\eta)-J(\bm\theta,\eta)\bigr|
\xrightarrow{p}0.
$

Profiling over $\eta$ preserves this uniform convergence: for every
$\bm\theta\in\Theta_{\mathcal A}$,
$
|J_{n,p}(\bm\theta)-J_p(\bm\theta)|
\le
\sup_{\eta\in\mathcal A_{\bm\theta}}
|J_n(\bm\theta,\eta)-J(\bm\theta,\eta)|
\le
\sup_{(\bm\theta,\eta)\in\mathcal A}
|J_n(\bm\theta,\eta)-J(\bm\theta,\eta)|,
$
and hence
$
\sup_{\bm\theta\in\Theta_{\mathcal A}}
|J_{n,p}(\bm\theta)-J_p(\bm\theta)|
\xrightarrow{p}0.
$

It remains to apply the profile identification condition in the theorem.
For any $\epsilon>0$, set
\(B_\epsilon=\{\bm\theta\in\Theta_{\mathcal A}:\|\bm\theta-\bm\theta^\star\|\ge\epsilon\}\).
If $B_\epsilon$ is empty, the claim is immediate. Otherwise, let
$
\Delta_\epsilon
:=
\inf_{\bm\theta\in B_\epsilon}J_p(\bm\theta)
-
J_p(\bm\theta^\star)
$
denote the identification gap, which is positive by assumption.
On the event
$
\sup_{\bm\theta\in\Theta_{\mathcal A}}
|J_{n,p}(\bm\theta)-J_p(\bm\theta)|
<\Delta_\epsilon/3,
$
no minimizer of $J_{n,p}$ can lie in $B_\epsilon$. Indeed, if
$\bm\theta\in B_\epsilon$, then
$
J_{n,p}(\bm\theta)
\ge
J_p(\bm\theta)-\Delta_\epsilon/3
\ge
J_p(\bm\theta^\star)+2\Delta_\epsilon/3
>
J_{n,p}(\bm\theta^\star).
$
Thus, on this event and on the event that the empirical minimizer exists,
$
\|\hat{\bm\theta}_n-\bm\theta^\star\|<\epsilon.
$
Lemma~\ref{lem:rpo-existence} gives existence with probability tending to one, and
the uniform convergence event also has probability tending to one. Hence
$
P(\|\hat{\bm\theta}_n-\bm\theta^\star\|\ge\epsilon)\to0
$
for every $\epsilon>0$, which proves
$
\hat{\bm\theta}_n\xrightarrow{p}\bm\theta^\star.
$
\end{proof}

\subsection{Proof of Proposition~\ref{prop:al-continuity}}

\begin{proof}
For any observation \((x_i,y_i)\), the definition of the soft-profiled residual gives
\[
\widetilde m^+_{\bm\theta,\eta}(x_i,y_i;u_i,\rho)
:=
\inf_{t\in[0,\eta]}
\left[
\{y_i-\psi_y(t;\bm\theta)\}^2
+
u_i\{\psi_x(t;\bm\theta)-x_i\}
+
{\rho}/{2}\{\psi_x(t;\bm\theta)-x_i\}^2
+
{u_i^2}/{2\rho}
\right].
\]
The function inside the infimum is continuous in \((t,\bm\theta,u_i)\) on the common compact set \([0,\eta_{\max}]\times\Theta\times[-U,U]\).
The feasible correspondence \(\eta\mapsto[0,\eta]\) is nonempty, compact-valued, and continuous in the Hausdorff distance on \([\eta_{\min},\eta_{\max}]\).
Therefore, by Berge's maximum theorem \citep[Theorem~17.31]{AliprantisBorder2006}, applied to the minimum problem through the negative of the objective, \(\widetilde m^+_{\bm\theta,\eta}(x_i,y_i;u_i,\rho)\) is continuous in
\((\bm\theta,\eta,u_i)\).

Since \(n\) is fixed, \(\widetilde R_n(\bm\theta,\eta;\bm u,\rho)\) is a finite average of continuous functions and is therefore continuous in \((\bm\theta,\eta,\bm u)\). Moreover, completing the square shows that \(\widetilde m^+_{\bm\theta,\eta}(x_i,y_i;u_i,\rho)\) is the infimum over \(t\in[0,\eta]\) of
$
\{y_i-\psi_y(t;\bm\theta)\}^2
+
{\rho}/{2}
\left\{
\psi_x(t;\bm\theta)-x_i+{u_i}/{\rho}
\right\}^2,
$
and is therefore nonnegative.
Hence \(\widetilde R_n(\bm\theta,\eta;\bm u,\rho)+\delta\ge \delta>0\), so the logarithm
is well defined and continuous.
It follows that \(\widetilde J_n(\bm\theta,\eta;\bm u,\rho)\) is continuous in \((\bm\theta,\eta,\bm u)\) on
\(\Theta\times[\eta_{\min},\eta_{\max}]\times[-U,U]^n\).
\end{proof}

\subsection{Proof of Proposition~\ref{prop:al-root-approx}}

\begin{proof}
Write \(\hat t_i=\hat t_i(\bm\theta,\eta)\) and \(r(t)=r(t;\bm\theta,x_i)\). Since
\(r(\hat t_i)=0\) and \(r'(\hat t_i)=\partial_t\psi_x(\hat t_i;\bm\theta)\neq0\), there exist
\(\delta>0\) and constants \(0<c_1\le c_2<\infty\) such that the interval
\(I_\delta\) with endpoints \(\hat t_i-\delta\) and \(\hat t_i+\delta\) is contained in \((0,\eta)\) and, for all \(t\in I_\delta\),
$
c_1|t-\hat t_i|\le |r(t)|\le c_2|t-\hat t_i|.
$
This follows from the continuity of \(r'\) and the mean-value theorem. Let \(g(t)=\|y_i-\psi_y(t;\bm\theta)\|^2\). Since \(\psi_y\) is continuously differentiable on
\(I_\delta\), \(g\) is Lipschitz on \(I_\delta\). Thus, for some \(L<\infty\),
\(|g(t)-g(\hat t_i)|\le L|t-\hat t_i|\) for all \(t\in I_\delta\).

For any \(U<\infty\) with \(|u_i|\le U\), let \(\widetilde t_i=\widetilde t_i(\bm\theta,\eta;u_i,\rho)\) be any minimizer over \(I_\delta\) of
$
g(t)+u_i r(t)+{\rho}/{2}r^2(t).
$
Since \(r(\hat t_i)=0\), the minimizing property gives
$
g(\widetilde t_i)+u_i r(\widetilde t_i)+{\rho}/{2}r^2(\widetilde t_i)
\le
g(\hat t_i).
$
Therefore,
$
{\rho}/{2}r^2(\widetilde t_i)
\le
g(\hat t_i)-g(\widetilde t_i)-u_i r(\widetilde t_i)
\le
L|\widetilde t_i-\hat t_i|+U|r(\widetilde t_i)|.
$
Using \(|r(\widetilde t_i)|\le c_2|\widetilde t_i-\hat t_i|\) and
\(|r(\widetilde t_i)|\ge c_1|\widetilde t_i-\hat t_i|\), we obtain
$
{\rho}/{2}c_1^2|\widetilde t_i-\hat t_i|^2
\le
(L+Uc_2)|\widetilde t_i-\hat t_i|.
$
If \(\widetilde t_i=\hat t_i\), the desired bound is immediate. Otherwise, dividing both sides
by \(|\widetilde t_i-\hat t_i|\) yields
$
|\widetilde t_i-\hat t_i|
\le
{2(L+Uc_2)}/{\rho c_1^2}.
$
Hence \(|\widetilde t_i(\bm\theta,\eta;u_i,\rho)-\hat t_i(\bm\theta,\eta)|=O(\rho^{-1})\), uniformly for \(|u_i|\le U\). This completes the proof.
\end{proof}

\subsection{Proof of Theorem~\ref{thm:sliced-prediction}}

\begin{proof}
Let \(f(t,\bm\theta)=\psi_x(t;\bm\theta)-x_0\).
Let $t_1^\star,\dots,t_K^\star$ be the latent roots specified in the theorem, so that
$
\psi_x(t_k^\star;\bm\theta^\star)=x_0
$
for $k=1,\dots,K$.

By assumption, $f(t_k^\star,\bm\theta^\star)=0$ and
\(\partial_t f(t_k^\star,\bm\theta^\star)=\partial_t\psi_x(t_k^\star;\bm\theta^\star)\neq0\)
for $k=1,\dots,K$. There exists $\gamma>0$ such that the intervals
$
I_k=[t_k^\star-\gamma,t_k^\star+\gamma],
~ k=1,\dots,K,
$
are pairwise disjoint and contained in $(0,\eta^\star)$, and such that
$
t_K^\star+\gamma<\eta^\star-\gamma
~\text{and}~
\gamma<\delta.
$

By the implicit function theorem, there exists a neighborhood of
$\bm\theta^\star$ on which, for each $k$, there is a unique continuous root branch
$t_k(\bm\theta)\in I_k$ satisfying
$
f(t_k(\bm\theta),\bm\theta)=0
$
and $t_k(\bm\theta^\star)=t_k^\star$.

It remains to rule out extra roots. Let
$
C
=
[0,\eta^\star+\delta]
\setminus
\bigcup_{k=1}^K(t_k^\star-\gamma/2,t_k^\star+\gamma/2)
$
be the complement of the root neighborhoods within \([0,\eta^\star+\delta]\).
Since the equation $f(t,\bm\theta^\star)=0$ has no solutions on $C$,
compactness gives
$
\inf_{t\in C}|f(t,\bm\theta^\star)|>0.
$
By continuity of $f$, there exists a (possibly smaller) neighborhood of
$\bm\theta^\star$ such that
$
f(t,\bm\theta)\neq0
$
for all $\bm\theta$ in this neighborhood and all $t\in C$. Hence, every root
of $f(t,\bm\theta)=0$ in $[0,\eta^\star+\delta]$ lies in one of the intervals
$I_1,\dots,I_K$, and there is exactly one such root in each interval.

Since $\hat{\bm\theta}_n\xrightarrow{p}\bm\theta^\star$ and
$\hat\eta_n\xrightarrow{p}\eta^\star$, with probability tending to one,
$\hat{\bm\theta}_n$ belongs to the above neighborhood and
$
|\hat\eta_n-\eta^\star|<\gamma.
$
On this event, $\hat\eta_n<\eta^\star+\delta$ and
$
\hat\eta_n>\eta^\star-\gamma>t_K^\star+\gamma.
$
Therefore all intervals $I_k$ are contained in $[0,\hat\eta_n]$, and the root
set of $\psi_x(t;\hat{\bm\theta}_n)=x_0$ over $[0,\hat\eta_n]$ is exactly
$
\{t_1(\hat{\bm\theta}_n),\dots,t_K(\hat{\bm\theta}_n)\}.
$

Thus, on an event whose probability tends to one, \(\mathcal F_{\hat{\bm\theta}_n,\hat\eta_n}(x_0)\) and \(\mathcal F_{\bm\theta^\star,\eta^\star}(x_0)\) reduce respectively to
\[
\left\{
\psi_y(t_k(\hat{\bm\theta}_n);\hat{\bm\theta}_n):k=1,\dots,K
\right\},
~
\left\{
\psi_y(t_k^\star;\bm\theta^\star):k=1,\dots,K
\right\}.
\]

Because $t_k(\bm\theta)$ is continuous at $\bm\theta^\star$ and $\psi_y$ is
continuous,
$$
\max_{1\le k\le K}
\left\|
\psi_y(t_k(\hat{\bm\theta}_n);\hat{\bm\theta}_n)
-
\psi_y(t_k^\star;\bm\theta^\star)
\right\|
\xrightarrow{p}0.
$$
The Hausdorff distance between the two finite sets is bounded above by this
maximum branchwise difference. Hence
$
d_H\{\mathcal F_{\hat{\bm\theta}_n,\hat\eta_n}(x_0),
\mathcal F_{\bm\theta^\star,\eta^\star}(x_0)\}
\xrightarrow{p}0.
$
\end{proof}

\subsection{Proof of Proposition~\ref{prop:eta-convergence}}

\begin{proof}
Let \(M_n=\max_{1\le i\le n}T_i\).
Let \(T\) denote a random variable with the common distribution of \(T_1,\ldots,T_n\).
Since the support of \(T\) is contained in \([0,\eta^\star]\), we have \(M_n\le \eta^\star\) almost surely.
Hence, for any \(\epsilon>0\), \(P(M_n>\eta^\star+\epsilon)=0\).
It remains to control the lower tail.
Because \([0,\eta^\star]\) is the support of \(T\), every left neighborhood of \(\eta^\star\) has positive probability. Therefore \(P(T>\eta^\star-\epsilon)>0\).
By independence,
\(P(M_n\le \eta^\star-\epsilon)=P(T\le \eta^\star-\epsilon)^n\).
Since \(P(T>\eta^\star-\epsilon)>0\), we have
\(P(T\le \eta^\star-\epsilon)<1\),
and therefore \(P(T\le \eta^\star-\epsilon)^n\to0\).
Combining the two bounds gives
$
P(|M_n-\eta^\star|>\epsilon)
\le
P(M_n\le\eta^\star-\epsilon)+P(M_n>\eta^\star+\epsilon)
\to0.
$
Thus \(M_n\xrightarrow{p}\eta^\star\).
\end{proof}

\subsection{Proof of Corollary~\ref{cor:sliced-inverse-prediction}}

\begin{proof}
The proof is analogous to that of Theorem~\ref{thm:sliced-prediction}, with the roles of
\(\psi_x\) and \(\psi_y\) interchanged. Let \(g(t,\bm\theta)=\psi_y(t;\bm\theta)-y_0\). By assumption,
\(g(s^\star_\ell,\bm\theta^\star)=0\) and
\(\partial_t g(s^\star_\ell,\bm\theta^\star)=\partial_t\psi_y(s^\star_\ell;\bm\theta^\star)\neq0\)
for \(\ell=1,\ldots,L\).

There exists \(\gamma>0\) such that the intervals
\(I_\ell=[s^\star_\ell-\gamma,s^\star_\ell+\gamma]\), \(\ell=1,\ldots,L\), are pairwise disjoint and
contained in \((0,\eta^\star)\). By the implicit function theorem, there exists a neighborhood
of \(\bm\theta^\star\) on which, for each \(\ell\), there is a unique continuous
root branch \(s_\ell(\bm\theta)\in I_\ell\) satisfying
\(g(s_\ell(\bm\theta),\bm\theta)=0\) and \(s_\ell(\bm\theta^\star)=s^\star_\ell\).

It remains to rule out extra roots. Let
$
\mathcal C
=
[0,\eta^\star+\delta]\setminus\bigcup_{\ell=1}^L(s^\star_\ell-\gamma/2,s^\star_\ell+\gamma/2)
$
be the complement of the root neighborhoods within \([0,\eta^\star+\delta]\).
Since \(g(t,\bm\theta^\star)\neq0\) on the compact set \(\mathcal C\), continuity gives
\(\inf_{t\in \mathcal C}|g(t,\bm\theta^\star)|>0\). Hence there exists a (possibly smaller) neighborhood
of \(\bm\theta^\star\) such that \(g(t,\bm\theta)\neq0\) for all \(t\in \mathcal C\) and all \(\bm\theta\) in this
neighborhood. Therefore all roots of \(\psi_y(t;\bm\theta)=y_0\) on
\([0,\eta^\star+\delta]\) lie in \(I_1,\ldots,I_L\), with exactly one root in each interval.

Since \(\hat{\bm\theta}_n\xrightarrow{p}\bm\theta^\star\) and
\(\hat\eta_n\xrightarrow{p}\eta^\star\), with probability tending to one,
\(\hat{\bm\theta}_n\) belongs to the above neighborhood and \([0,\hat\eta_n]\) contains all
the intervals \(I_1,\ldots,I_L\) while remaining inside \([0,\eta^\star+\delta]\). On this
event, \(\mathcal G_{\hat{\bm\theta}_n,\hat\eta_n}(y_0)\) and \(\mathcal G_{\bm\theta^\star,\eta^\star}(y_0)\) reduce respectively to
$
\{\psi_x(s_\ell(\hat{\bm\theta}_n);\hat{\bm\theta}_n):\ell=1,\ldots,L\},
~
\{\psi_x(s^\star_\ell;\bm\theta^\star):\ell=1,\ldots,L\}.
$
Because \(s_\ell(\bm\theta)\) is continuous at \(\bm\theta^\star\) and \(\psi_x\) is continuous,
$
\max_{1\le \ell\le L}
\|\psi_x(s_\ell(\hat{\bm\theta}_n);\hat{\bm\theta}_n)
-
\psi_x(s^\star_\ell;\bm\theta^\star)\|\allowbreak
\xrightarrow{p}0.
$
The Hausdorff distance between the two finite sets is bounded above by this maximum
branchwise difference. Hence
\[
d_H\{\mathcal G_{\hat{\bm\theta}_n,\hat\eta_n}(y_0),\mathcal G_{\bm\theta^\star,\eta^\star}(y_0)\}
\xrightarrow{p}0.
\]
\end{proof}

\bibliographystyle{abbrvnat}
\bibliography{refs}

@book{Bertsekas1999,
  author    = {Bertsekas, Dimitri P.},
  title     = {Nonlinear Programming},
  edition   = {2},
  publisher = {Athena Scientific},
  address   = {Belmont, MA},
  year      = {1999}
}

@article{wang2025augmentation,
  author  = {Wang, Shulei},
  title   = {Augmentation Invariant Manifold Learning},
  journal = {Journal of the Royal Statistical Society Series B:
             Statistical Methodology},
  year    = {2025},
  volume  = {87},
  number  = {4},
  pages   = {978--1000},
  doi     = {10.1093/jrsssb/qkaf003}
}

@article{dunson2026inferring,
  author  = {Dunson, David B. and Wu, Nan},
  title   = {Inferring Manifolds Using Gaussian Processes},
  journal = {Biometrika},
  year    = {2026},
  volume  = {113},
  number  = {2},
  eid     = {asag011},
  doi     = {10.1093/biomet/asag011}
}

@article{KeglKrzyzakLinderZeger2000,
  author  = {K{\'e}gl, Bal{\'a}zs and Krzy{\.z}ak, Adam and Linder, Tam{\'a}s and Zeger, Kenneth},
  title   = {Learning and Design of Principal Curves},
  journal = {IEEE Transactions on Pattern Analysis and Machine Intelligence},
  year    = {2000},
  volume  = {22},
  number  = {3},
  pages   = {281--297},
  doi     = {10.1109/34.841759}
}

@book{AliprantisBorder2006,
  author    = {Aliprantis, Charalambos D. and Border, Kim C.},
  title     = {Infinite Dimensional Analysis: A Hitchhiker's Guide},
  edition   = {3},
  publisher = {Springer},
  address   = {Berlin},
  year      = {2006},
  doi       = {10.1007/3-540-29587-9}
}

@article{HastieStuetzle1989,
  author  = {Hastie, Trevor and Stuetzle, Werner},
  title   = {Principal Curves},
  journal = {Journal of the American Statistical Association},
  volume  = {84},
  number  = {406},
  pages   = {502--516},
  year    = {1989},
  doi     = {https://doi.org/10.1080/01621459.1989.10478797}
}

@Manual{R-optim,
  title        = {R: A Language and Environment for Statistical Computing},
  author       = {{R Core Team}},
  organization = {R Foundation for Statistical Computing},
  address      = {Vienna, Austria},
  year         = {2026},
  url          = {https://www.R-project.org/}
}

@Manual{R-nloptr,
  title  = {nloptr: R Interface to NLopt},
  author = {Ypma, Jelmer and Johnson, Steven G. and Stamm, Aymeric and Borchers, Hans W. and Eddelbuettel, Dirk and Ripley, Brian and Hornik, Kurt and Chiquet, Julien and Adler, Avraham and Dai, Xiongtao and Ooms, Jeroen and Kalibera, Tomas and Jagan, Mikael},
  year   = {2025},
  note   = {R package version 2.2.1},
  url    = {https://CRAN.R-project.org/package=nloptr}
}

@article{LiMolchanovMolinariPeng2021,
  author  = {Li, Qiyu and Molchanov, Ilya and Molinari, Francesca and Peng, Sida},
  title   = {Local Regression Smoothers with Set-Valued Outcome Data},
  journal = {International Journal of Approximate Reasoning},
  volume  = {128},
  pages   = {129--150},
  year    = {2021},
  doi     = {https://doi.org/10.1016/j.ijar.2020.10.005}
}

@article{lissajous1857memoire,
  title   = {M{\'e}moire sur l'{\'e}tude optique des mouvements vibratoires},
  author  = {Lissajous, Jules Antoine},
  journal = {Annales de Chimie et de Physique},
  series  = {3},
  volume  = {51},
  pages   = {147--231},
  year    = {1857}
}

@book{mayergoyz2003mathematical,
  title     = {Mathematical Models of Hysteresis and Their Applications},
  author    = {Mayergoyz, Isaak D.},
  year      = {2003},
  publisher = {Elsevier},
  address   = {Amsterdam}
}

@book{molchanov2005random,
  title     = {Theory of Random Sets},
  author    = {Molchanov, Ilya},
  year      = {2005},
  publisher = {Springer},
  address   = {London}
}

@article{murphy2000profile,
  author  = {Murphy, Susan A. and van der Vaart, Aad W.},
  title   = {On Profile Likelihood},
  journal = {Journal of the American Statistical Association},
  year    = {2000},
  volume  = {95},
  number  = {450},
  pages   = {449--465},
  doi     = {10.1080/01621459.2000.10474219}
}

@book{NocedalWright2006,
  author    = {Nocedal, Jorge and Wright, Stephen J.},
  title     = {Numerical Optimization},
  edition   = {2},
  series    = {Springer Series in Operations Research and Financial Engineering},
  publisher = {Springer},
  address   = {New York},
  year      = {2006},
  doi       = {10.1007/978-0-387-40065-5}
}

@article{Stuart2010,
  author  = {Stuart, A. M.},
  title   = {Inverse Problems: A {B}ayesian Perspective},
  journal = {Acta Numerica},
  volume  = {19},
  pages   = {451--559},
  year    = {2010},
  doi     = {https://doi.org/10.1017/S0962492910000061}
}

@Manual{Soetaert2009RootSolve,
  author = {Soetaert, Karline},
  title  = {{rootSolve}: Nonlinear root finding, equilibrium and steady-state analysis of ordinary differential equations},
  year   = {2009},
  note   = {R package version 1.6},
  url    = {https://CRAN.R-project.org/package=rootSolve}
}

@book{vandervaart1998asymptotic,
  title     = {Asymptotic Statistics},
  author    = {van der Vaart, Aad W.},
  year      = {1998},
  publisher = {Cambridge University Press},
  address   = {Cambridge}
}

@misc{weisstein2026trochoid,
  title        = {Trochoid},
  author       = {Weisstein, Eric W.},
  howpublished = {From MathWorld--A Wolfram Resource},
  year         = {2026},
  note         = {Available at \url{https://mathworld.wolfram.com/Trochoid.html}}
}

@book{yates1952handbook,
  title     = {A Handbook on Curves and Their Properties},
  author    = {Yates, Robert C.},
  year      = {1952},
  publisher = {J. W. Edwards},
  address   = {Ann Arbor, MI}
}

@book{spong2005robot,
  title={Robot Modeling and Control},
  author={Spong, Mark W. and Hutchinson, Seth and Vidyasagar, M.},
  year={2005},
  publisher={Wiley},
  address={Hoboken, NJ}
}

@inproceedings{zhu2017toward,
  title={Toward Multimodal Image-to-Image Translation},
  author={Zhu, Jun-Yan and Zhang, Richard and Pathak, Deepak and Darrell, Trevor and Efros, Alexei A. and Wang, Oliver and Shechtman, Eli},
  booktitle={Advances in Neural Information Processing Systems},
  volume={30},
  year={2017}
}

@inproceedings{rombach2022high,
  title={High-Resolution Image Synthesis with Latent Diffusion Models},
  author={Rombach, Robin and Blattmann, Andreas and Lorenz, Dominik and Esser, Patrick and Ommer, Bj{\"o}rn},
  booktitle={Proceedings of the IEEE/CVF Conference on Computer Vision and Pattern Recognition},
  pages={10684--10695},
  year={2022}
}

\end{document}